\documentclass[12pt,onecolumn,journal,draftclsnofoot,doublespace]{IEEEtran}
\usepackage{epsf,psfrag,amssymb,amsfonts,amsmath,graphicx,cite,subfigure}

\def\boxit#1{\vbox{\hrule\hbox{\vrule\kern3pt
        \vbox{\kern3pt#1\kern3pt}\kern3pt\vrule}\hrule}}

\def\reals{ { {\rm  I \kern-0.15em R }  } }
\def\complex{ {\,{{\rm C} \kern-0.50em \raise0.20ex {  |}}\, }}

\def\Rbf{{\bf R}}

\def\Fc{{\cal F}}

\def\Kc{{\cal K}}

\def\Pc{{\cal P}}

\def\Rc{{\cal R}}

\def\Tc{{\cal T}}

\def\Zc{{\cal Z}}

\def\be{\begin{equation}}
\def\ee{\end{equation}}

\def\defeq{{\stackrel{\Delta}{=}}}
\def\scalefig#1{\epsfxsize #1\textwidth}
\def\Rxx{\Rbf_{\ssstyle X\kern-.1em X}}

\let\ssstyle=\scriptscriptstyle

\def\etal{{\it et al. \/}}

\def\ie{{\it i.e.,\ \/}}
\def\Kout{\setbox1=\hbox{\Huge\bf K}\hbox to
1.05\wd1{\hspace{.05\wd1}
}}
\def\Sout{\setbox1=\hbox{\Huge\bf S}\hbox to 1.05\wd1{\hspace{.05\wd1}
}}

\def\ie{{\it i.e.,\ \/}}
\def\etc{{\it etc.}}
\def\defeq{{\stackrel{\Delta}{=}}}

\def\scalefig#1{\epsfxsize #1\textwidth}
\def\nn{{\nonumber}}
\newcommand{\mbbE}{\mathbb{E}}
\newtheorem{lemma}{Lemma}
\newtheorem{theorem}{Theorem}

\newtheorem{definition}{Definition}
\newtheorem{corollary}{Corollary}

\begin{document}
\title{\LARGE Indexability of Restless Bandit Problems and Optimality
of Whittle's Index for Dynamic Multichannel Access}
\author{ Keqin Liu, ~~~ Qing Zhao\\
 University of California, Davis, CA 95616\\
 kqliu@ucdavis.edu, qzhao@ece.ucdavis.edu }

\maketitle


\begin{abstract}
We  \footnotetext{This work was supported by the
Army Research Laboratory CTA on Communication and Networks under
Grant DAAD19-01-2-0011 and by the National Science Foundation under
Grants ECS-0622200 and CCF-0830685.\\ Part of this work was
presented at the 5th IEEE Conference on Sensor, Mesh and Ad Hoc
Communications and Networks (SECON) Workshops (June, 2008) and the
IEEE Asilomar Conference on Signals, Systems, and Computers
(October, 2008).}
consider a class of restless multi-armed bandit problems (RMBP) that arises
in dynamic multichannel access, user/server scheduling, and optimal activation
in multi-agent systems.  For this class of
RMBP, we establish the indexability and obtain Whittle's index in
closed-form for both discounted and average reward criteria. These
results lead to a direct implementation of Whittle's index policy
with remarkably low complexity. When these Markov chains are stochastically
identical, we show that Whittle's index policy is optimal under
certain conditions. Furthermore, it has a semi-universal structure
that obviates the need to know the Markov transition probabilities.
The optimality and the semi-universal structure result from the
equivalency between Whittle's index policy and the myopic policy
established in this work. For non-identical channels, we develop
efficient algorithms for computing  a performance upper bound
given by Lagrangian relaxation. The tightness of the upper
bound and the near-optimal performance of Whittle's index policy are
illustrated with simulation examples.

\end{abstract}

\vspace{-1em}

\begin{IEEEkeywords}
\vspace{-1em} Opportunistic access, dynamic channel selection,
restless multi-armed bandit, Whittle's index, indexability, myopic
policy.
\end{IEEEkeywords}

\section{Introduction}\label{sec:inc}

\subsection{Restless Multi-armed Bandit
Problem}\label{subsec:bandit}

Restless Multi-armed Bandit Process (RMBP) is a generalization of the classical Multi-armed Bandit
Processes (MBP), which has been studied since
1930's~\cite{Mahajan&Teneketzis:07bookchapter}. In an MBP, a player,
with full knowledge of the current state of each arm, chooses
\emph{one} out of $N$ arms to activate at each time and receives a reward
determined by the state of the activated arm. Only the activated arm
changes its state according to a Markovian rule while the states of
passive arms are frozen. The objective is to maximize the long-run
reward over the infinite horizon by choosing which arm to activate
at each time.

The structure of the optimal policy for the classical
MBP was established by Gittins in 1979~\cite{Gittins}, who proved
that an index policy is optimal. The significance of Gittins' result
is that it reduces the complexity of finding the optimal policy for
an MBP from exponential with $N$ to linear with $N$. Specifically,
an index policy assigns an index to each state of each arm and
activates the arm whose current state has the largest index. Arms
are decoupled when computing the index, thus reducing an
$N-$dimensional problem to $N$ independent $1-$dimensional problems.

Whittle generalized MBP to RMBP by allowing multiple $(K\ge 1)$ arms
to be activated simultaneously and allowing passive arms to also
change states~\cite{whittle}. Either of these two generalizations
would render Gittins' index policy suboptimal in general, and
finding the optimal solution to a general RMBP has been shown to be
PSPACE-hard by Papadimitriou and Tsitsiklis~\cite{Complexity}. In fact, merely allowing multiple
plays $(K\ge 1)$ would have fundamentally changed the problem as
shown in the classic work by Anantharam \etal
\cite{Ananthram&etal:87TAC} and by Pandelis and Teneketzis
\cite{Pandelis&Teneketzis:99}.

By considering the Lagrangian relaxation of the problem, Whittle
proposed a heuristic index policy for RMBP~\cite{whittle}. Whittle's
index policy is the optimal solution to RMBP under a relaxed
constraint: the number of activated arms can vary over time provided
that its average over the infinite horizon equals to $K$. This
average constraint leads to decoupling among arms, subsequently, the
optimality of an index policy. Under the strict constraint that
exactly $K$ arms are to be activated at each time, Whittle's index
policy has been shown to be asymptotically
optimal under certain conditions ($N\rightarrow\infty$ stochastically identical
arms)~\cite{Weber}. In the finite regime,
extensive empirical studies have demonstrated its near-optimal
performance, see, for example, \cite{Ansell},~\cite{Glazebrook1}.

The difficulty of Whittle's index policy lies in the complexity of
establishing its existence and computing the index, especially for
RMBP with uncountable state space as in our case. Not every RMBP
has a well-defined Whittle's index; those that admit Whittle's index
policy are called {\it indexable}~\cite{whittle}. The indexability
of an RMBP is often difficult to establish, and computing Whittle's
index can be complex, often relying on numerical approximations.

In this paper, we show that for a significant class
of RMBP most relevant to multichannel
dynamic access applications, the indexability
can be established and Whittle's index can be obtained in closed form.
For stochastically identical arms, we establish the equivalency between
Whittle's index policy and the myopic policy. This result, coupled with
recent findings in \cite{Zhao&etal:08TWC,Ahmad&etal} on the myopic policy for this class of RMBP, shows that
Whittle's index policy achieves the optimal performance under certain conditions
and has a semi-universal structure that is robust against model mismatch and
variations. This class of RMBP is described next.

\subsection{Dynamic Multichannel Access}

Consider the problem of probing $N$ independent Markov chains.
Each chain has two states---``good'' and ``bad''--- with different
transition probabilities across chains (see Fig.~\ref{fig:MC}). At each time,
a player can choose $K~(1\le K< N)$ chains to probe and receives
reward determined by the states of the probed chains. The objective is
to design an optimal policy that governs the selection of $K$ chains
at each time to maximize the long-run reward.

\begin{figure}[htb]
\centerline{
\begin{psfrags}
\psfrag{A}[c]{ $0$} \psfrag{B}[c]{ $1$} \psfrag{A1}[c]{\small (bad)}
\psfrag{B1}[c]{\small (good)} \psfrag{a}[c]{ $p^{(i)}_{01}$}
\psfrag{b}[l]{ $p^{(i)}_{11}$} \psfrag{a1}[r]{ $p^{(i)}_{00}$}
\psfrag{b1}[c]{ $p^{(i)}_{10}$}
\scalefig{0.5}\epsfbox{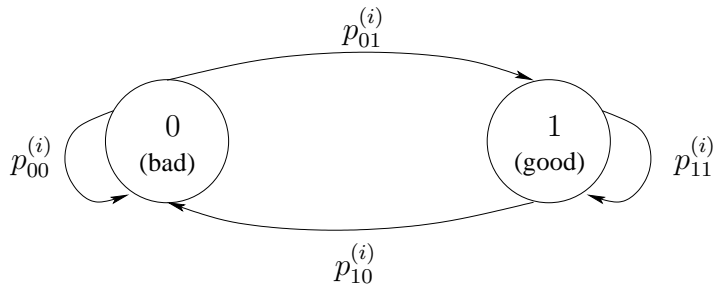}
\end{psfrags}}
\caption{The Gilber-Elliot channel model.} \label{fig:MC}
\end{figure}

The above general problem arises in a wide range of communication
systems, including cognitive radio networks, downlink scheduling in
cellular systems, opportunistic transmission over fading channels,
and resource-constrained jamming and anti-jamming.
In the communications context, the $N$ independent Markov chains corresponds
to $N$ communication channels under the Gilbert-Elliot channel
model~\cite{Gilbert:60}, which has been commonly used to abstract
physical channels with memory (see, for example,
\cite{Zorzi&etal:98COM},~\cite{Johnston&Krishnamurthy:06TWC}).
The state of a channel models the communication quality of this channel and
determines the resultant reward of accessing this channel.
For example, in
cognitive radio networks where secondary users search in the
spectrum for idle channels temporarily unused by primary
users~\cite{Zhao&Sadler:07SPM}, the state of a channel models the
occupancy of the channel. For downlink scheduling in cellular
systems, the user is a base station, and each channel is associated
with a downlink mobile receiver. Downlink receiver scheduling is
thus equivalent to channel selection.

The application of this
problem also goes beyond communication systems. For example, it has
applications in target tracking as considered in~\cite{Leny08ACC},
where $K$ unmanned aerial vehicles are tracking the states of
$N~(N>K)$ targets in each slot.

\subsection{Main Results}\label{subsec:contribution}

Fundamental questions concerning Whittle's index policy since the
day of its invention have been its existence, its performance, and
the complexity in computing the index. What are the necessary and/or
sufficient conditions on the state transition and the reward
structure that make an RMBP indexable? When can Whittle's index be
obtained in closed-form? For which special classes of RMBP is
Whittle's index policy optimal? When numerical evaluation has to be
resorted to in studying its performance, are there easily computable
performance benchmarks?

In this paper, we attempt to address these questions for the class of RMBP
described above. As will be shown, this class of RMBP has an \emph{uncountable} state space,
making the problem highly nontrivial. The underlying two-state
Markov chain that governs the state transition of each arm, however,
brings rich structures into the problem, leading to positive and
surprising answers to the above questions. The wide range of
applications of this class of RMBP makes the results obtained in
this paper generally applicable.

Under both discounted and average reward criteria, we establish the
indexability of this class of RMBP. The basic technique of our
proof is to bound the total amount of time that an arm is made
passive under the optimal policy. The general approach of using the
total passive time in proving indexability was considered by Whittle
in~\cite{whittle} when showing that a classic MBP is always
indexable. Applying this approach to a nontrivial RMBP is, however,
much more involved, and our proof appears to be the first that
extends this approach to RMBP. We hope that this work contributes to
the set of possible techniques for establishing indexability of RMBP.

Based on the indexability, we show that
Whittle's index can be obtained in closed-form for both discounted
and average reward criteria. This result reduces the complexity of
implementing Whittle's index policy to simple evaluations of these
closed-form expressions. This result is particularly significant considering
the uncountable state space which would render numerical approaches
impractical. The monotonically increasing and piecewise concave (for
arms with $p_{11}\ge p_{01}$) or piecewise convex (for arms with
$p_{11}<p_{01}$) properties of Whittle's index are also established.
The monotonicity of Whittle's index leads to an interesting
equivalency with the myopic policy --- the simplest
nontrivial index policy --- when arms are stochastically identical.
This equivalency allows us to work on the myopic index,
which has a much simpler form, when establishing the structure and
optimality of Whittle's index policy for stochastically identical
arms.

As to the performance of Whittle's index policy for this class of
RMBP, we show that under certain conditions, Whittle's index policy
is optimal for stochastically identical arms. This result provides examples for the
optimality of Whittle's index policy in the finite regime. The
approximation factor of Whittle's index policy (the ratio of the
performance of Whittle's index policy to that of the optimal policy)
is analyzed when the optimality conditions do not hold.
Specifically, we show that when arms are stochastically
identical, the approximation factor of Whittle's index policy is at
least $\frac{K}{N}$ when $p_{11}\ge p_{01}$ and at least
$\max\{\frac{1}{2},\frac{K}{N}\}$ when $p_{11}<p_{01}$.

When arms are non-identical, we develop an efficient
algorithm to compute a performance upper bound based on Lagrangian relaxation.
We show that this algorithm runs in at most
$O(N(\log N)^2)$ time to compute the performance upper
bound within $\epsilon$-accuracy for any $\epsilon>0$. Furthermore,
when every channel satisfies $p_{11}<p_{01}$, we can compute the
upper bound without error with complexity $O(N^2\log N)$.


Another interesting finding is that when arms are stochastically
identical, Whittle's index policy has a semi-universal structure
that obviates the need to know the Markov transition probabilities.
The only required knowledge about the Markovian model is the order of
$p_{11}$ and $p_{01}$. This semi-universal structure reveals the
robustness of Whittle's index policy against model mismatch and
variations.

\subsection{Related Work}\label{subsec:Relatedwork}

Multichannel opportunistic access in the context of cognitive radio
systems has been studied in~\cite{Zhao&etal:07JSAC,Chen&etal:08IT}
where the problem is formulated as a Partially Observable Markov
Decision Process (POMDP) to take into account potential correlations
among channels. For stochastically identical and independent
channels and under the assumption of single-channel sensing ($K=1$),
the structure, optimality, and performance of the myopic policy have
been investigated in~\cite{Zhao&etal:08TWC}, where the
semi-universal structure of the myopic policy was established for
all $N$ and the optimality of the myopic policy proved for $N=2$. In
a recent work~\cite{Ahmad&etal}, the optimality of the myopic policy
was extended to $N>2$ under the condition of $p_{11}\ge p_{01}$. These results
have also been extended to cases with probing errors in \cite{Zhao&Krish:08TACsub}.
In this paper, we establish the equivalence relationship between the
myopic policy and Whittle's index policy when channels are
stochastically identical. This equivalency shows that
the results obtained in~\cite{Zhao&etal:08TWC,Ahmad&etal} for the
myopic policy are directly applicable to Whittle's index policy.
Furthermore, we extend these results to multichannel sensing
($K>1$).

Other examples of applying the general RMBP framework to communication
systems include the work by Lott and
Teneketzis~\cite{Lott&Teneketzis:00} and the work by Raghunathan
\etal \cite{Raghunathan&etal:08INFOCOM}.
In~\cite{Lott&Teneketzis:00}, the problem of multichannel allocation
for single-hop mobile networks with multiple service classes was
formulated as an RMBP, and sufficient conditions for the optimality
of a myopic-type index policy were established. In
\cite{Raghunathan&etal:08INFOCOM}, multicast scheduling in wireless
broadcast systems with strict deadlines was formulated as an RMBP
with a finite state space. The indexability was established and
Whittle's index was obtained in closed-form. Recent work by Kleinberg
gives interesting applications of bandit processes to Internet
search and web advertisement placement~\cite{Kleinberg:08}.

In the general context of RMBP, there is a rich literature on
indexability. See~\cite{Nino} for the linear programming
representation of conditions for indexability and \cite{Glazebrook1}
for examples of specific indexable restless bandit processes.
Constant-factor approximation algorithms for RMBP have also been
explored in the literature. For the same class of RMBP as considered
in this paper, Guha and Munagala~\cite{Guha} have developed a
constant-factor (1/68) approximation via LP relaxation under the
condition of $p_{11}>\frac{1}{2}>p_{01}$ for each channel.
In~\cite{Guha1}, Guha \etal have developed a factor-2 approximation
policy via LP relaxation for the so-called monotone bandit
processes.

In~\cite{Leny08ACC}, Le Ny \etal have considered the same class of RMBP motivated by the
applications of target tracking. They have
independently established the indexability and obtained the
closed-form expressions for Whittle's index under the discounted
reward criterion\footnote{A conference version of our result was
published in June, 2008, the same time as~\cite{Leny08ACC}.}.
Our approach to establishing indexability and obtaining Whittle's index is, however,
different from that used in \cite{Leny08ACC}, and the two approaches complement
each other. Indeed, the fact that two completely different applications
lead to the same class of RMBP lends support
for a detailed investigation of this particular
type of RMBP. We also include several results
that were not considered in \cite{Leny08ACC}.  In particular, we
consider both discounted and average reward criterion, develop
algorithms for and analyze the complexity of computing the optimal
performance under the Lagrangian relaxation, and establish the
semi-universal structure and the optimality of Whittle's index
policy for stochastically identical arms.

\subsection{Organization}\label{subsec:organization}

The rest of the paper is organized as follows. In
Sec.~\ref{sec:formulation}, the RMBP formulation is presented. In
Sec.~\ref{sec:indexpolicy}, we introduce the basic concepts of
indexability and Whittle's index. In
Sec.~\ref{sec:discount}, we address the total discounted reward
criterion, where we establish the indexability, obtain Whittle's
index in closed-form, and develop efficient algorithms for computing an upper bound
on the performance of the optimal policy. Simulation examples are
provided to illustrate the tightness of the upper bound and the
near-optimal performance of Whittle's index policy. In
Sec.~\ref{sec:whittleindexpolicyavg}, we consider the average reward
criterion and obtain results parallel to those obtained under the
discounted reward criterion. In Sec.~\ref{sec:identical}, we
consider the special case when channels are stochastically
identical. We show that Whittle's index policy is optimal under certain
conditions and has a simple and
robust structure. The approximation factor of Whittle's
index policy is also analyzed. Sec.~\ref{sec:conclusion} concludes this paper.

\section{Problem Statement and Restless Bandit
Formulation}\label{sec:formulation}

\subsection{Multi-channel Opportunistic Access}
Consider $N$ independent Gilbert-Elliot channels, each with
transmission rate $B_i (i=1,\cdots,N)$. Without loss of generality,
we normalize the maximum data rate:
$\max_{i\in\{1,2,\cdots,N\}}\{B_i\}=1$. The state of channel
$i$---``good''($1$) or ``bad''($0$)--- evolves from slot to slot as
a Markov chain with transition matrix ${\bf
P}_i=\{p_{j,k}^{(i)}\}_{j,k\in\{0,1\}}$ as shown in
Fig.~\ref{fig:MC}.

At the beginning of slot $t$, the user selects $K$ out of $N$
channels to sense. If the state $S_i(t)$ of the sensed channel $i$
is $1$, the user transmits and collects $B_i$ units of reward in
this channel. Otherwise, the user collects no reward in this
channel. Let $U(t)$ denote the set of $k$ channels chosen in slot
$t$. The reward obtained in slot $t$ is thus given by
\[R_{U(t)}(t)=\Sigma_{i\in U(t)}S_i(t)B_i.\] Our objective is to
maximize the expected long-run reward by designing a sensing policy
that sequentially selects $K$ channels to sense in each slot.

\subsection{Restless Multi-armed Bandit Formulation}\label{subsec:banditformula}

The channel states $[S_1(t),...,S_N(t)]\in\{0,1\}^N$ are not
directly observable before the sensing action is made. The user can,
however, infer the channel states from its decision and observation
history. It has been shown that a sufficient statistic for optimal
decision making is given by the conditional probability that each
channel is in state $1$ given all past decisions and observations
\cite{Smallwood&Sondik:71OR}. Referred to as the belief vector or
information state, this sufficient statistic is denoted by
$\Omega(t)\,\defeq\, [\omega_1(t),\cdots,\omega_N(t)]$, where
$\omega_i(t)$ is the conditional probability that $S_i(t)=1$. Given
the sensing action $U(t)$ and the observation in slot $t$, the
belief state in slot $t+1$ can be obtained recursively as follows:

{\begin{equation} \omega_i(t+1)=\left\{\begin{array}{ll}
p^{(i)}_{11}, &  i\in U(t), S_{i}(t)=1\\
p^{(i)}_{01}, &  i\in U(t), S_{i}(t)=0\\
\Tc(\omega_i(t)), & i\notin U(t)\\
\end{array}
\right., \label{eq:omega}
\end{equation}}
where
\[\Tc(\omega_i(t))\triangleq\omega_i(t)p_{11}^{(i)}+(1-\omega_i(t))p_{01}^{(i)}\]
denotes the operator for the one-step belief update for unobserved
channels.

If no information on the initial system state is available, the
$i$-th entry of the initial belief vector $\Omega(1)$ can be set to
the stationary distribution $\omega^{(i)}_o$ of the underlying
Markov chain:

\begin{equation}
\omega^{(i)}_o=\frac{p^{(i)}_{01}}{p^{(i)}_{01}+p^{(i)}_{10}}.
\label{eq:wo}
\end{equation}

It is now easy to see that we have an RMBP, where each channel is
considered as an arm and the state of arm $i$ in slot $t$ is the
belief state $\omega_i(t)$. The user chooses an action $U(t)$
consisting of $K$ arms to activate (sense) in each slot, while other
arms are made passive (unobserved). The states of both active and
passive arms change as given in~\eqref{eq:omega}.

A policy $\pi:\Omega(t)\rightarrow U(t)$ is a function that maps
from the belief vector $\Omega(t)$ to the action $U(t)$ in slot $t$.
Our objective is to design the optimal policy $\pi^*$ to maximize
the expected long-term reward.

There are two commonly used performance measures. One is the
expected total \emph{discounted} reward over the infinite horizon:
\begin{eqnarray}\label{eqn:measure1}
\mathbb{E}_{\pi}[\Sigma_{t=1}^{\infty}\beta^{t-1}R_{\pi(\Omega(t))}(t)|\Omega(1)],
\end{eqnarray}
where $0\le\beta<1$ is the discount factor and
$R_{\pi(\Omega(t))}(t)$ is the reward obtained in slot $t$ under
action $U(t)=\pi(\Omega(t))$ determined by the policy $\pi$. This
performance measure applies when rewards in the future are less
valuable, for example, in delay sensitive communication systems. It
also applies when the horizon length is a geometrically distributed
random variable with parameter $\beta$. For example, a communication
session may end at a random time, and the user aims to maximize the
number of packets delivered before the session ends.

The other performance measure is the expected \emph{average} reward
over the infinite horizon~\cite{arapostathis}:
\begin{eqnarray}\label{eqn:measure2}
\mathbb{E}_{\pi}[\lim_{T\rightarrow\infty}\frac{1}{T}\Sigma_{t=1}^{T}R_{\pi(\Omega(t))}(t)|\Omega(1)].
\end{eqnarray}
This is the common measure of throughput in the context of
communications.
%

For notation convenience, let $(\Omega(1),\{{\bf
P}_i\}_{i=1}^N,\{B_i\}_{i=1}^N,\beta)$ denote the RMBP
with the discounted reward criterion, and $(\Omega(1),\{{\bf
P}_i\}_{i=1}^N,\{B_i\}_{i=1}^N,1)$ the RMBP with the
average reward criterion.

\section{Indexability and Index Policies}\label{sec:indexpolicy}
In this section, we introduce the basic concepts of
indexability and Whittle's index policy.

\subsection{Index Policy}\label{subsec:indexpolicy}

An \emph{index policy} assigns an index for each state of each arm
to measure how rewarding it is to activate an arm at a particular
state. In each slot, the policy activates those $K$ arms whose
current states have the largest indices.

For a strongly decomposable index policy, the index of an arm only
depends on the characteristics (transition probabilities, reward
structure, \etc) of this arm. Arms are thus decoupled when computing
the index, reducing an $N-$dimensional problem to $N$ independent
$1-$dimensional problems.

A myopic policy is a simple example of strongly decomposable index
policies. This policy ignores the impact of the current action on
the future reward, focusing solely on maximizing the expected
immediate reward. The index is thus the expected immediate reward of
activating an arm at a particular state. For the problem at hand,
the myopic index of each state $\omega_i(t)$ of arm $i$ is simply
$\omega_i(t)B_i$. The myopic action $\hat{U}(t)$ under the belief
state $\Omega(t) = [\omega_1(t), \cdots, \omega_N(t)]$ is given by
\begin{equation}
\hat{U}(t)=\arg\max_{U(t)} \Sigma_{i\in U(t)} \, \omega_i(t)B_i.
\label{eq:a*}
\end{equation}


\subsection{Indexability and Whittle's Index Policy}\label{subsec:whittleindexpolicy}

To introduce indexability and Whittle's index, it suffices to
consider a single arm due to the strong decomposability of Whittle's
index. Consider a single-armed bandit process (a single channel)
with transition probabilities $\{p_{j,k}\}_{j,k\in{0,1}}$ and
bandwidth $B$ (here we drop the channel index for notation
simplicity). In each slot, the user chooses one of two possible
actions---$u\in\{0~(\mbox{passive}),~1~(\mbox{active})\}$---to make
the arm passive or active. An expected reward of $\omega B$ is
obtained when the arm is activated at belief state $\omega$, and the
belief state transits according to~\eqref{eq:omega}. The objective
is to decide whether to active the arm in each slot to maximize the
total discounted or average reward. The optimal policy is
essentially given by an optimal partition of the state space $[0,1]$
into a passive set $\{\omega:u^*(\omega)=0\}$ and an active set
$\{\omega:u^*(\omega)=1\}$, where $u^*(\omega)$ denotes the optimal
action under belief state $\omega$.

Whittle's index measures how attractive it is to activate an arm
based on the concept of {\it subsidy for passivity}. Specifically,
we construct a single-armed bandit process that is identical to the
above specified bandit process except that a constant subsidy $m$ is
obtained whenever the arm is made passive. Obviously, this subsidy
$m$ will change the optimal partition of the passive and active
sets, and states that remain in the active set under a larger
subsidy $m$ are more attractive to the user. The minimum subsidy $m$
that is needed to move a state from the active set to the passive
set under the optimal partition thus measures how attractive this
state is.

We now present the formal definition of indexability and Whittle's index. We
consider the discounted reward criterion. Their definitions under
the average reward criterion can be similarly obtained.

Denoted by $V_{\beta,m}(\omega)$, the value function represents the
maximum expected total discounted reward that can be accrued from a
single-armed bandit process with subsidy $m$ when the initial belief
state is $\omega$. Considering the two possible actions in the first
slot, we have
\begin{eqnarray}\label{eqn:value}
V_{\beta,m}(\omega)=\max\{V_{\beta,m}(\omega;u=0),~V_{\beta,m}(\omega;u=1)\},
\end{eqnarray}
where $V_{\beta,m}(\omega;u)$ denotes the expected total discounted
reward obtained by taking action $u$ in the first slot followed by the
optimal policy in future slots. Consider $V_{\beta,m}(\omega;u=0)$.
It is given by the sum of the subsidy $m$ obtained in the first slot
under the passive action and the total discounted future reward
$\beta V_{\beta,m}(\Tc(\omega))$ which is determined by the updated belief
state $\Tc(\omega)$ (see~\eqref{eq:omega}).
$V_{\beta,m}(\omega;u=1)$ can be similarly obtained, and we arrive
at the following dynamic programming.
\begin{eqnarray}\label{eqn:value_a}
V_{\beta,m}(\omega;u=0)&=&m+\beta
V_{\beta,m}(\Tc(\omega)),\\\label{eqn:value_b}
V_{\beta,m}(\omega;u=1)&=&\omega+\beta(\omega
V_{\beta,m}(p_{11})+(1-\omega)V_{\beta,m}(p_{01})).
\end{eqnarray}
The optimal action $u_m^*(\omega)$ for belief state $\omega$ under subsidy $m$ is
given by
\begin{eqnarray}
u_m^*(\omega)=\left\{\begin{array}{l}1,~~~\mbox{if}~V_{\beta,m}(\omega;u=1)>V_{\beta,m}(\omega;u=0)\\
0,~~~\mbox{otherwise}\end{array}\right..
\end{eqnarray}
The passive set $\Pc(m)$ under subsidy $m$ is given by
\begin{eqnarray}\label{eqn:passiveset}
\Pc(m)&=&\{\omega:~u^*_m(\omega)=0\}\\\label{eqn:passiveset1}
&=&\{\omega:~V_{\beta,m}(\omega;u=0)\ge V_{\beta,m}(\omega;u=1)\}
\end{eqnarray}

\begin{definition}
An arm is \emph{indexable} if the passive set $\Pc(m)$ of the
corresponding single-armed bandit process with subsidy $m$
monotonically increases from $\emptyset$ to the whole state space
$[0,1]$ as $m$ increases from $-\infty$ to $+\infty$. An RMBP is
indexable if every arm is indexable.
\end{definition}

Under the indexability condition, Whittle's index is defined as
follows.

\begin{definition}\label{def:whittleindex}
If an arm is indexable, its \emph{Whittle's index} $W(\omega)$ of
the state $\omega$ is the infimum subsidy $m$ such that it is
optimal to make the arm passive at $\omega$. Equivalently, Whittle's
index $W(\omega)$ is the infimum subsidy $m$ that makes the passive
and active actions equally rewarding.
\begin{eqnarray}
W(\omega)&=&\inf_{m}\{m:~u^*_m(\omega)=0\}\\
&=&\inf_{m}\{m:~V_{\beta,m}(\omega;u=0)=V_{\beta,m}(\omega;u=1)\}.
\end{eqnarray}
\end{definition}
%
%
%

In Fig.~\ref{fig:throughput}, we compare the performance
(throughput) of the myopic policy, Whittle's index policy, and the
optimal policy for the RMBP formulated in
Sec.~\ref{sec:formulation}. We observe that Whittle's index policy
achieves a near-optimal performance while the myopic policy suffers
from a significant performance loss.

\begin{figure}[h]
\centerline{
\begin{psfrags}
\scalefig{0.6}\epsfbox{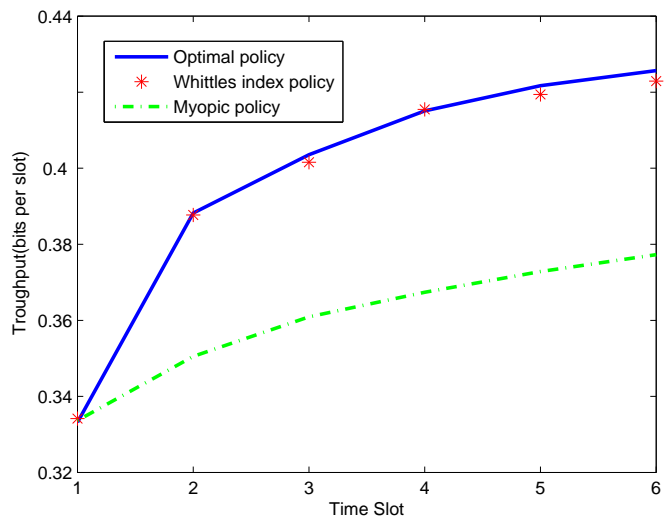}
\end{psfrags}}
\caption{The performance by Whittle's index policy
($K=1,~N=7,~\{p^{(i)}_{01}\}_{i=1}^7=\{0.8,0.6,0.4,0.9,0.8,0.6,
0.7\},~\{p^{(i)}_{11}\}_{i=1}^7=\{0.6,0.4,0.2,0.2,0.4,0.1
,0.3\},~\mbox{and}~B_i=\{0.4998,0.6668,1.0000,0.6296,0.5830
,0.8334,0.6668\}$).} \label{fig:throughput}
\end{figure}

\section{Whittle's index under discounted reward criterion}\label{sec:discount}

In this section, we focus on the discounted reward criterion. We
establish the indexability, obtain Whittle's index in closed-form,
and develop efficient algorithms for computing an upper bound of the
optimal performance to provide a benchmark for evaluating the
performance of Whittle's index policy.

\subsection{Properties of Belief State Transition}\label{subsec:modeldiscount}

To establish indexability and obtain Whittle's index, it suffices to
consider the single-armed bandit process with subsidy $m$. Again, we
drop the channel index from all notations and set $B=1$.

\begin{figure}[h]
\centerline{
\begin{psfrags}
\psfrag{a}[r]{ $\omega$} \psfrag{k}[c]{ $k$} \psfrag{p}[l]{
$\Tc^k(\omega)$}\psfrag{w}[l]{$\omega_o$}
\scalefig{0.8}\epsfbox{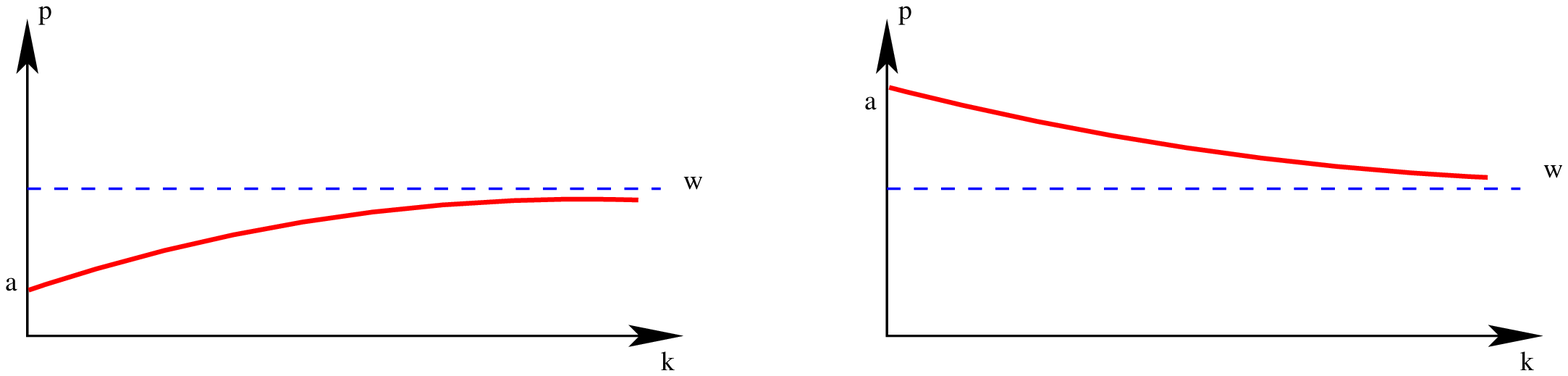}
\end{psfrags}}
\caption{The $k$-step belief update of an unobserved arm $(p_{11}\ge
p_{01}$).} \label{fig:ksteppositive}
\end{figure}

\begin{figure}[h]
\centerline{
\begin{psfrags}
\psfrag{w}[r]{ $\omega$}\psfrag{a}[r]{} \psfrag{b}[r]{}
\psfrag{k}[c]{ $k$} \psfrag{p}[l]{
$\Tc^k(\omega)$}\psfrag{wo}[l]{$\omega_o$}\psfrag{0}[c]{$0$}
\psfrag{1}[c]{$1$}\psfrag{2}[c]{$2$}\psfrag{3}[c]{$3$}
\scalefig{0.8}\epsfbox{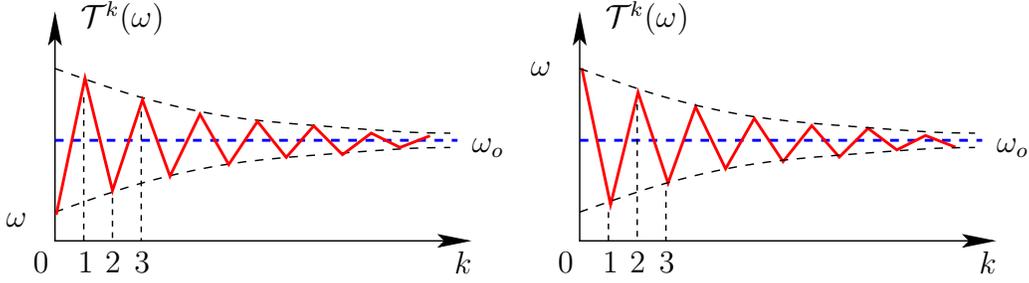}
\end{psfrags}}
\caption{The $k$-step belief update of an unobserved arm
($p_{11}<p_{01}$).} \label{fig:kstepnegative}
\end{figure}

The following lemma establishes properties of belief state
transition that reveal the basic structure of the RMBP considered in
this paper. We resort often to these properties when deriving the
main results.
\begin{lemma}\label{lemma:ksteptransition}
Let $\Tc^k(\omega(t))\defeq \Pr[S(t+k)=1|\omega(t)]~(k=0,1,2,\cdots)$
denote the $k-$step belief update of $\omega(t)$ when the arm is
unobserved for $k$ consecutive slots. We have
\begin{eqnarray}\label{eqn:beliefupdate}
&&\Tc^k(\omega)=\frac{p_{01}-(p_{11}-p_{01})^k(p_{01}-(1+p_{01}-p_{11})\omega)}{1+p_{01}-p_{11}},\\[0.5em]
&&\min\{p_{01},p_{11}\}\le\Tc^k(\omega)\le\max\{p_{01},p_{11}\},~~\forall~\omega\in[0,1],~\forall~k\ge1.
\end{eqnarray}
Furthermore, the convergence of $\Tc^k(\omega)$ to the stationary
distribution $\omega_o=\frac{p_{01}}{p_{01}+p_{10}}$ has the
following property.
\begin{itemize}
\item \emph{Case 1: Positively correlated channel\footnote{
It is easy to show that $p_{11} > p_{01}$ corresponds to the case
where the channel states in two consecutive slots are positively
correlated, \ie for any distribution of $S(t)$, we have
$\mbbE[(S(t)-\mbbE[S(t)])(S(t+1)-\mbbE[S(t+1)])]>0$, where $S(t)$ is
the state of the Gilbert-Elliot channel in slot $t$. Similar,
$p_{11} < p_{01}$ corresponds to the case where $S(t)$ and $S(t+1)$
are negatively correlated, and $p_{11}=p_{01}$ the case where $S(t)$
and $S(t+1)$ are independent.} ($p_{11}\ge p_{01}$).}\\
For any $\omega\in[0,1]$, $\Tc^k(\omega)$ monotonically converges to
$\omega_o$ as $k\rightarrow\infty$ (see Fig.~\ref{fig:ksteppositive}).\\[-1em]
\item \emph{Case 2: Negatively correlated channel ($p_{11}< p_{01}$).}\\
For any $\omega\in[0,1]$, $\Tc^{2k}(\omega)$ and
$\Tc^{2k+1}(\omega)$ converge, from opposite directions, to
$\omega_o$ as $k\rightarrow\infty$ (see
Fig.~\ref{fig:kstepnegative}).
\end{itemize}
\end{lemma}

\vspace{0.5em}

\begin{proof}
$\Tc^k(\omega)=\omega\Tc^k(1)+(1-\omega)\Tc^k(0)$, where
$\Tc^k(1)=\Pr[S(t+k)=1|S(t)=1]$ is the $k-$step transition
probability from $1$ to $1$, and $\Tc^k(0)=\Pr[S(t+k)=1|S(t)=0]$ is
the $k-$step transition probability from $0$ to $1$. From the
eigen-decomposition of the transition matrix $\bf{P}$ (see
\cite{Gallager:95book}), we have
$\Tc^k(1)=\frac{p_{01}+(1-p_{11})(p_{11}-p_{01})^k}{1+p_{01}-p_{11}}$
and $\Tc^k(0)=\frac{p_{01}(1-(p_{11}-p_{01})^k)}{1+p_{01}-p_{11}}$,
which leads to~\eqref{eqn:beliefupdate}. Other properties follow
directly from~\eqref{eqn:beliefupdate}.
\end{proof}

\vspace{1em}

Next, we define an important quantity $L(\omega,\omega')$. Referred
to as the {\it crossing time}, $L(\omega,\omega')$ is the minimum
amount of time required for a passive arm to transit across
$\omega'$ starting from $\omega$.
\[L(\omega,\omega')\defeq\min\{k:~\Tc^k(\omega)>\omega'\}.\]
For a positively correlated arm, we have, from
Lemma~\ref{lemma:ksteptransition},
\begin{eqnarray}\label{eqn:Lpositive}
L(\omega,\omega')=\left\{\begin{array}{ll} 0,~~
&\mbox{if}~\omega>\omega'\\
\lfloor\log_{p_{11}-p_{01}}^{\frac{p_{01}-\omega'(1-p_{11}+p_{01})}{p_{01-\omega(1-p_{11}+p_{01})}}}\rfloor+1,~~
&\mbox{if}~\omega\le\omega'<\omega_o\\
\infty,~~&\mbox{if}~\omega\le\omega'~\mbox{and}~\omega'\ge\omega_o
\end{array}\right..
\end{eqnarray}
For a negatively correlated arm, we have\\
\begin{eqnarray}\label{eqn:Lnegative}
L(\omega,\omega')=\left\{\begin{array}{ll} 0,~~
&\mbox{if}~\omega>\omega'\\
1,~~ &\mbox{if}~\omega\le\omega'~\mbox{and}~\Tc(\omega)>\omega'\\
\infty,~~&\mbox{if}~\omega\le\omega'~\mbox{and}~\Tc(\omega)\le\omega'\end{array}\right..
\end{eqnarray}

\subsection{The Optimal Policy}\label{subsec:vpdiscount}

In this subsection, we show that the optimal policy for the
single-armed bandit process with subsidy $m$ is a threshold policy.
This threshold structure provides the key to establishing the
indexability and solving for Whittle's index policy in closed-form
as shown in Sec.~\ref{subsec:indexabilitydiscount}.

This threshold structure is obtained by examining the value
functions $V_{\beta,m}(\omega;u=0)$ and $V_{\beta,m}(\omega;u=1)$
given in~\eqref{eqn:value_a} and~\eqref{eqn:value_b}.
From~\eqref{eqn:value_b}, we observe that $V_{\beta,m}(\omega;u=1)$
is a linear function of $\omega$. Following the general result on
the convexity of the value function of a POMDP~\cite{Sondik}, we
conclude that $V_{\beta,m}(\omega;u=0)$ given in~\eqref{eqn:value_a}
is convex in $\omega$. These properties of $V_{\beta,m}(\omega;u=1)$
and $V_{\beta,m}(\omega;u=0)$ lead to the lemma below.

\begin{lemma}\label{lemma:thresholdPolicy} The optimal policy for the single-armed bandit process with subsidy $m$ is a
threshold policy, \ie there exists an $\omega_{\beta}^*(m)\in\mathbb{R}$
such that
\begin{eqnarray}
u_m^*(\omega)=\left\{\begin{array}{cc}1~~&\mbox{if}~\omega>\omega_{\beta}^*(m)\\
0~~&\mbox{if}~\omega\le\omega_{\beta}^*(m)\end{array}\right.,\nn
\end{eqnarray}
and $V_{\beta,m}(\omega_{\beta}^*(m);u=0)=V_{\beta,m}(\omega_{\beta}^*(m);u=1)$.
\end{lemma}

\begin{figure}[th]
\centerline{
\begin{psfrags}
\psfrag{0}[c]{$0$}\psfrag{1}[c]{$1$}
\psfrag{g}[c]{$V_{\beta,m}(\omega;u=1)$ }\psfrag{g1}[c]{}
\psfrag{h}[c]{ }\psfrag{h1}[c]{~$V_{\beta,m}(\omega;u=0)$}
\psfrag{wt}[c]{$\omega_{\beta}^*(m)$} \psfrag{p}[c]{ Passive Set}
\psfrag{a}[c]{ Active Set}
\psfrag{s}[c]{~$\omega<\omega_{\beta}^*(m)$} \psfrag{l}[c]{
$\omega>\omega_{\beta}^*(m)$}\psfrag{w}[c]{$\omega$}
\scalefig{0.6}\epsfbox{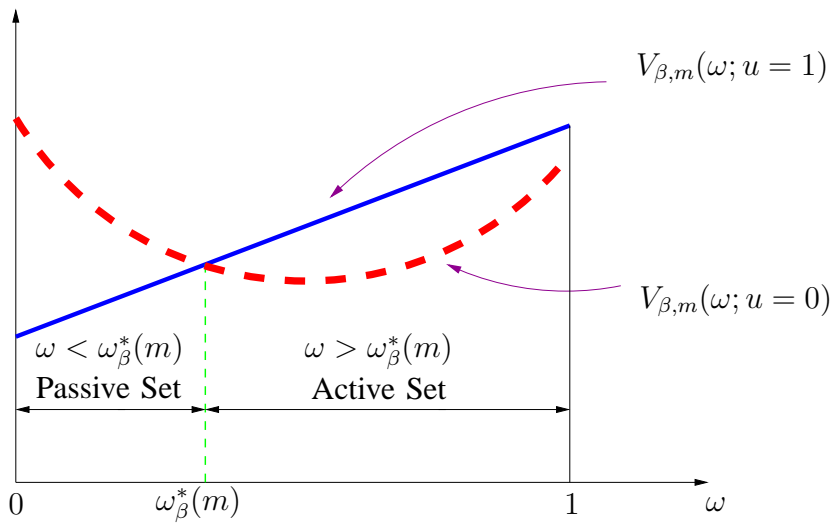}
\end{psfrags}}\caption{The optimality of a threshold policy ($0\le m<1).$}\label{fig:valuefunc}
\end{figure}

\begin{proof}
Consider first $0\le m<1$. We have the following inequality
regarding the end points of $V_{\beta,m}(0;u=1)$ and
$V_{\beta,m}(0;u=0)$ (see Fig.~\ref{fig:valuefunc}).
\begin{eqnarray}
V_{\beta,m}(0;u=1)&=&\beta V_{\beta,m}(p_{01})\le m+\beta
V_{\beta,m}(p_{01})=V_{\beta,m}(0;u=0),\\
V_{\beta,m}(1;u=1)&=&1+\beta V_{\beta,m}(p_{11})>m+\beta
V_{\beta,m}(p_{11})=V_{\beta,m}(1;u=0).
\end{eqnarray}
Since $V_{\beta,m}(\omega;u=1)$ is linear in $\omega$ and
$V_{\beta,m}(\omega;u=0)$ is convex in $\omega$,
$V_{\beta,m}(\omega;u=1)$ and $V_{\beta,m}(\omega;u=0)$ must have
one unique intersection at some point $\omega_{\beta}^*(m)$ as shown
in Fig.~\ref{fig:valuefunc}.

When $m\ge 1$, it is optimal to make the arm passive
all the time since the expected immediate reward $\omega$ by
activating the arm is uniformly upper bounded by $1$ (see
Fig.~\ref{fig:valuefunc2}). We can thus choose
$\omega_{\beta}^*(m)=c~\mbox{for any}~c>1$.

When $m<0$, we have (see Fig.~\ref{fig:valuefunc1})
\begin{eqnarray}
V_{\beta,m}(0;u=1)&=&\beta V_{\beta,m}(p_{01})> m+\beta
V_{\beta,m}(p_{01})=V_{\beta,m}(0;u=0),\\
V_{\beta,m}(1;u=1)&=&1+\beta V_{\beta,m}(p_{11})>m+\beta
V_{\beta,m}(p_{11})=V_{\beta,m}(0;u=0).
\end{eqnarray}
Based on the convexity of $V_{\beta,m}(\omega;u=0)$ in $\omega$, we
have $V_{\beta,m}(\omega;u=1)>V_{\beta,m}(\omega;u=0)$ for any
$\omega\in[0,1]$. It is thus optimal to always activate the arm, and
we can choose $\omega_{\beta}^*(m)=b~\mbox{for any}~b<0$.
Lemma~\ref{lemma:thresholdPolicy} thus follows. The expressions of
$V_{\beta,m}(0;u=1)$ and $V_{\beta,m}(0;u=0)$ given in
Fig.~\ref{fig:valuefunc2} and Fig.~\ref{fig:valuefunc1} are obtained
from the closed-form expression of the value function, which will be
shown in the next subsection.
\begin{figure}
\begin{minipage}{4in}
\hspace{-3em}\centerline{
\begin{psfrags}
\psfrag{0}[c]{$0$}\psfrag{1}[c]{$1$} \psfrag{w}[c]{$\omega$}
\psfrag{g}[c]{
$V_{\beta,m}(\omega;u=1)=\omega+\frac{m\beta}{1-\beta}$
}\psfrag{h}[c]{$V_{\beta,m}(\omega;u=0)=\frac{m}{1-\beta}$}
\scalefig{0.85}\epsfbox{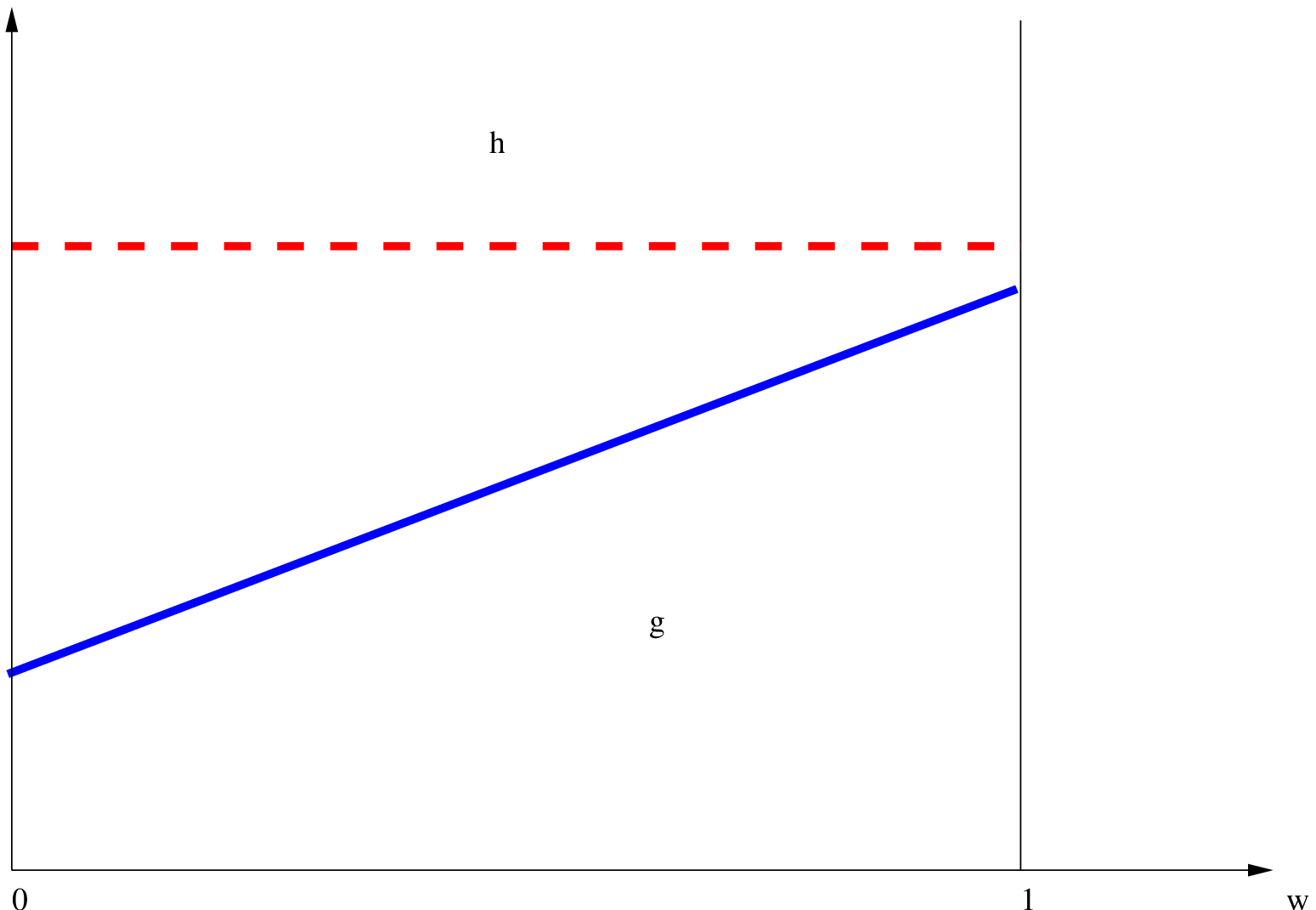}
\end{psfrags}}\caption{The optimality of a threshold policy ($m\ge1$).}\label{fig:valuefunc2}
\end{minipage}
\begin{minipage}{4in}
\hspace{-3em}\centerline{
\begin{psfrags}
\psfrag{0}[c]{$0$}\psfrag{1}[c]{$1$} \psfrag{w}[c]{$\omega$}
\psfrag{g}[c]{\small~~~~~~~~~~~~~~~
$V_{\beta,m}(\omega;u=1)=\frac{\omega(1-\beta)+p_{01}\beta}{(1-\beta)(1-\beta
p_{01}+\beta p_{01})}$
}\psfrag{h}[c]{\small~~~~~~~~~~~~~~~~~~~$V_{\beta,m}(\omega;u=0)=m+\beta\frac{\Tc^1(\omega)(1-\beta)+p_{01}\beta}{(1-\beta)(1-\beta
p_{01}+\beta p_{01})}$} \scalefig{0.85}\epsfbox{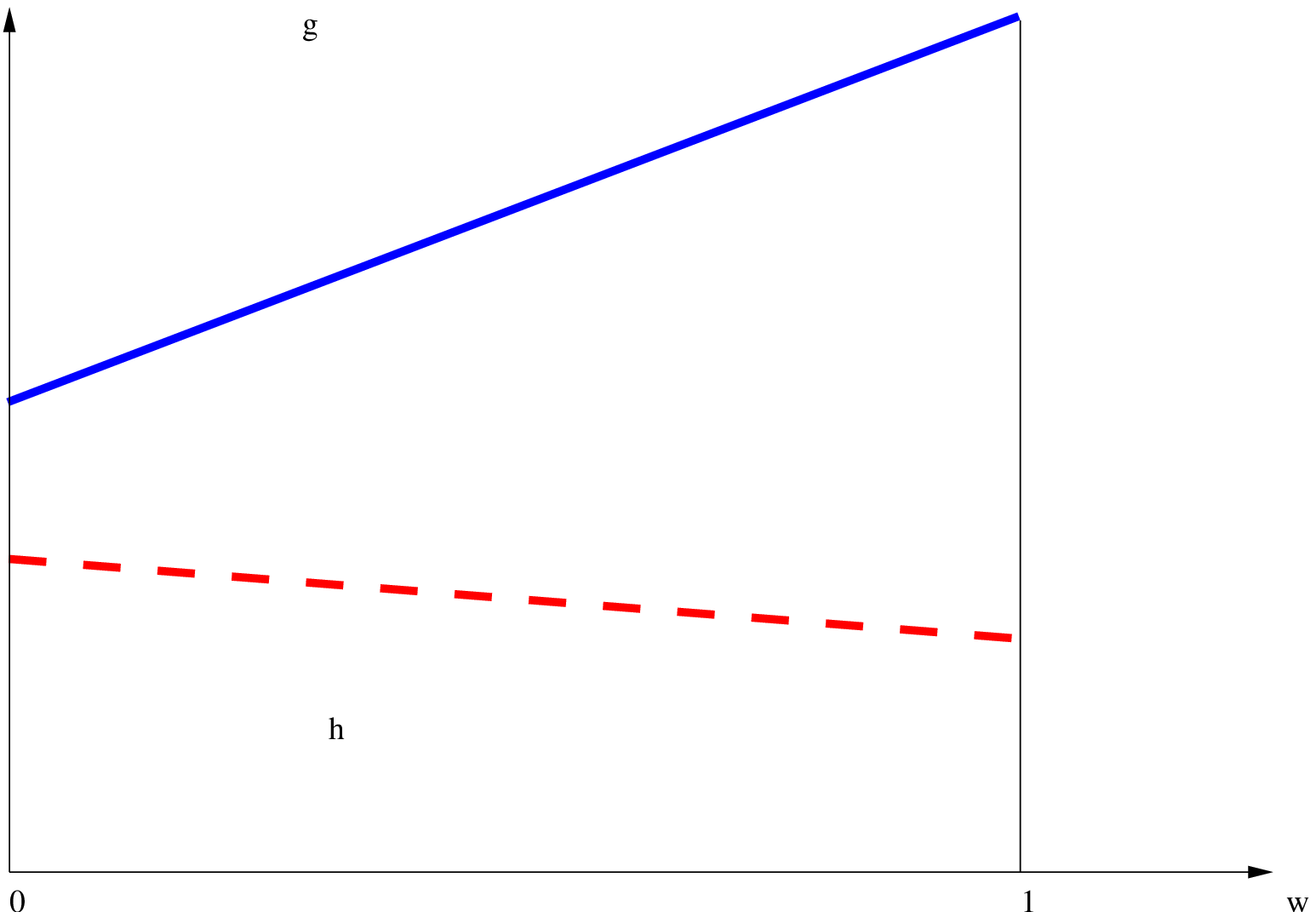}
\end{psfrags}}\caption{The optimality of a threshold policy ($m<0$.)}\label{fig:valuefunc1}
\end{minipage}
\end{figure}
\end{proof}

\subsection{Closed-form Expression of The Value
Function}\label{subsec:closeV}

In this subsection, we obtain closed-form expressions for the value
function $V_{\beta,m}(\omega)$. This result is fundamental to
calculating Whittle's index in closed-form and analyzing the
performance of Whittle's index policy.

Based on the threshold structure of the optimal policy, the value
function $V_{\beta,m}(\omega)$ can be expressed in terms of
$V_{\beta,m}(\Tc^k(\omega);u=1)$ for some
$t_0\in\Zc^+\cup\{\infty\}$, where
$t_0=L(\omega,\omega_{\beta}^*(m))+1$ is the index of the slot when
the belief $\omega$ transits across the threshold $\omega^*_{\beta}(m)$ for the
first time (recall that $L(\omega,\omega_{\beta}^*(m))$ is the crossing time
given in~\eqref{eqn:Lpositive} and~\eqref{eqn:Lnegative}).
Specifically, in the first $L(\omega,\omega_{\beta}^*(m))$ slots,
the subsidy $m$ is obtained in each slot. In slot
$t_0=L(\omega,\omega_{\beta}^*(m))+1$, the belief state transits
across the threshold $\omega_{\beta}^*(m)$ and the arm is activated.
The total reward thereafter is
$V_{\beta,m}(\Tc^{L(\omega,\omega_{\beta}^*(m))}(\omega);u=1)$. We
thus have, considering the discount factor,
\begin{eqnarray}
V_{\beta,m}(\omega)=
\frac{1-\beta^{L(\omega,\omega_{\beta}^*(m))}}{1-\beta}m+\beta^{L(\omega,\omega_{\beta}^*(m))}V_{\beta,m}
(\Tc^{L(\omega,\omega_{\beta}^*(m))}(\omega);u=1).\label{eq:vcloseform}
\end{eqnarray}
Since $V_{\beta,m}(\Tc^k(\omega);u=1)$ is a function of
$V_{\beta,m}(p_{01})$ and $V_{\beta,m}(p_{11})$ as shown
in~\eqref{eqn:value_a}, we only need to solve for
$V_{\beta,m}(p_{01})$ and $V_{\beta,m}(p_{11})$. Note that $p_{01}$
and $p_{11}$ are simply two specific values of $\omega$; both
$V_{\beta,m}(p_{01})$ and $V_{\beta,m}(p_{11})$ can be written as
functions of themselves through~\eqref{eq:vcloseform}. We can thus
solve for $V_{\beta,m}(p_{01})$ and $V_{\beta,m}(p_{11})$ as given
in Lemma~\ref{lemma:valuefunction}.
\begin{lemma}\label{lemma:valuefunction}
Let $\omega_{\beta}^*(m)$ denote the threshold of the optimal policy
for the single-armed bandit process with subsidy $m$. The value
functions $V_{\beta,m}(p_{01})$ and $V_{\beta,m}(p_{11})$ can be
obtained in closed-form as given below.
\begin{itemize}
\item {\em Case 1: Positively correlated channel ($p_{11}\ge p_{01}$)}
\end{itemize}
\begin{eqnarray}\label{eq:PcloseformVp01}
V_{\beta,m}(p_{01})&=&\left\{\begin{array}{ll}
\frac{p_{01}}{(1-\beta)(1-\beta p_{11}+\beta p_{01})}, & \mbox{if}~
\omega^*_{\beta}(m)<p_{01}\\
\frac{(1-\beta
p_{11})(1-\beta^{L(p_{01},\omega_{\beta}^*(m))})m+(1-\beta)\beta^{L(p_{01},\omega_{\beta}^*(m))}\Tc^{L(p_{01},\omega_{\beta}^*(m))}(p_{01})}
{(1-\beta
p_{11})(1-\beta)(1-\beta^{L(p_{01},\omega_{\beta}^*(m))+1})+(1-\beta)^2\beta^{L(p_{01},\omega_{\beta}^*(m))+1}\Tc^{L(p_{01},\omega_{\beta}^*(m))}(p_{01})},
& \mbox{if}~ p_{01}\le\omega^*_{\beta}(m)<\omega_o\\
\frac{m}{1-\beta}, & \mbox{if}~ \omega^*_{\beta}(m)\ge\omega_o
\end{array} \right.~~~\label{eq:PcloseformVp11}\\[0.5em]
V_{\beta,m}(p_{11})&=&\left\{\begin{array}{ll}
\frac{p_{11}+\beta(1-p_{11})V_{\beta,m}(p_{01})}{1-\beta p_{11}}, &
\mbox{if}~\omega^*_{\beta}(m)<p_{11}\\
\frac{m}{1-\beta}, & \mbox{if}~ \omega^*_{\beta}(m)\ge p_{11}
\end{array} \right..
\end{eqnarray}
Note that $V_{\beta,m}(p_{01})$ is given explicitly
in~\eqref{eq:PcloseformVp01} while $V_{\beta,m}(p_{11})$ is given in
terms of $V_{\beta,m}(p_{01})$ for the ease of presentation.
\begin{itemize}

\newpage

\item {\em Case 2: Negatively correlated channel ($p_{11}< p_{01}$)}
\end{itemize}
\begin{eqnarray}\label{eq:NcloseformVp11}
V_{\beta,m}(p_{11})&=&\left\{\begin{array}{ll}
\frac{p_{11}(1-\beta)+\beta p_{01}}{(1-\beta)(1-\beta p_{11}+\beta
p_{01})}, & \mbox{if}~\omega^*_{\beta}(m)<p_{11}\\
\frac{m(1-\beta(1-p_{01}))+\beta\Tc(p_{11})(1-\beta)+\beta^2p_{01}}
{1-\beta(1-p_{01})-\beta^2\Tc(p_{11})(1-\beta)-\beta^3p_{01}}, &
\mbox{if}~
p_{11}\le\omega^*_{\beta}(m)<\Tc(p_{11})\\
\frac{m}{1-\beta}, & \mbox{if}~ \omega^*_{\beta}(m)\ge\Tc(p_{11})
\end{array} \right..\label{eq:NcloseformVp01}\\[0.5em]
V_{\beta,m}(p_{01})&=&\left\{\begin{array}{ll} \frac{p_{01}+\beta
p_{01}V_{\beta,m}(p_{11})}{1-\beta(1-p_{01})}, & \mbox{if}~
\omega^*_{\beta}(m)<p_{01}\\
\frac{m}{1-\beta}, & \mbox{if}~ \omega^*_{\beta}(m)\ge p_{01}
\end{array} \right..
\end{eqnarray}
Note that $V_{\beta,m}(p_{11})$ is given explicitly
in~\eqref{eq:NcloseformVp11} while $V_{\beta,m}(p_{01})$ is given in
terms of $V_{\beta,m}(p_{11})$ for the ease of presentation.
\end{lemma}

\begin{proof}
The key to the closed-form expressions for $V_{\beta,m}(p_{01})$ and
$V_{\beta,m}(p_{11})$ is finding the first slot that the optimal
action is to activate the arm (\ie the belief state transits across
the threshold $\omega^*_{\beta}(m)$). This can be done by applying
the transition properties of the belief state given in
Lemma~\ref{lemma:ksteptransition}. See Appendix A for the complete
proof.
\end{proof}

\subsection{The Total Discounted Time of Being Passive}

In this subsection, we study the total discounted time that the
single-armed bandit process with subsidy $m$ is made passive. This
quantity plays the central role in our proof of indexability and in
the algorithms of computing an upper bound of the optimal
performance as shown in Sec.~\ref{subsec:indexabilitydiscount} and
Sec.~\ref{subsec:performancediscount}.

Let $D_{\beta,m}(\omega)$ denote the total discounted time that the
single-armed bandit process with subsidy $m$ is made passive under
the optimal policy when the initial belief state is $\omega$. It has
been shown by Whittle that $D_{\beta,m}(\omega)$ is the derivative
of the value function $V_{\beta,m}(\omega)$ with respect to
$m$~\cite{whittle}:
\[D_{\beta,m}(\omega)=\frac{d(V_{\beta,m}(\omega))}{dm}.\]
This result is intuitive: when the subsidy for passivity $m$
increases, the rate at which the total discounted reward
$V_{\beta,m}(\omega)$ increases is determined by how often the arm
is made passive.

Based on the threshold structure of the optimal policy, we can
obtain the following dynamic programming equation for
$D_{\beta,m}(\omega)$ similar to that for $V_{\beta,m}(\omega)$
given in~\eqref{eq:vcloseform}.
\begin{eqnarray}\label{eqn:closeformD}
D_{\beta,m}(\omega)=
\frac{1-\beta^{L(\omega,\omega^*_{\beta}(m))}}{1-\beta}+\beta^{L(\omega,\omega^*_{\beta}(m))+1}(\Tc^{L(\omega,\omega^*_{\beta}(m))}(\omega)
D_{\beta,m}(p_{11})+(1-\Tc^{L(\omega,\omega^*_{\beta}(m))}(\omega))D_{\beta,m}(p_{01})).
\label{eq:dcloseform}
\end{eqnarray}
Specifically, the first term in~\eqref{eqn:closeformD} is the total
discounted time of the first $L(\omega,\omega^*_{\beta}(m))$ slots
when the arm is made passive. In slot
$L(\omega,\omega^*_{\beta}(m))+1$, the arm is activated. With
probability $\Tc^{L(\omega,\omega^*_{\beta}(m))}(\omega)$, the
channel is in the good state in this slot, and the total future discounted
passive time is
$D_{\beta,m}(p_{11})$. With probability
$1-\Tc^{L(\omega,\omega^*_{\beta}(m))}(\omega)$, the channel is in
the bad state in this slot, and the total future discounted passive time
is $D_{\beta,m}(p_{01})$.

By considering $\omega=p_{01}$ and $\omega=p_{11}$, both
$D_{\beta,m}(p_{01})$ and $D_{\beta,m}(p_{11})$ can be written as
functions of themselves through~\eqref{eq:dcloseform}. We can thus
solve for $D_{\beta,m}(p_{01})$ and $D_{\beta,m}(p_{11})$ as given
in Lemma~\ref{lemma:closeformD}.
\begin{lemma}\label{lemma:closeformD}
Let $\omega_{\beta}^*(m)$ denote the threshold of the optimal policy
for the single-armed bandit process with subsidy $m$. The total
discounted passive times $D_{\beta,m}(p_{01})$ and $D_{\beta,m}(p_{11})$ are
given as follows.
\begin{itemize}
\item {\em Case 1: Positively correlated channel ($p_{11}\ge p_{01}$)}
\end{itemize}
\begin{eqnarray}
D_{\beta,m}(p_{01})&=&\left\{\begin{array}{ll} 0, & \mbox{if}~
\omega^*_{\beta}(m)<p_{01}\\
\frac{(1-\beta p_{11})(1-\beta^{L(p_{01},\omega_{\beta}^*(m))})}
{(1-\beta
p_{11})(1-\beta)(1-\beta^{L(p_{01},\omega_{\beta}^*(m))+1})+(1-\beta)^2\beta^{L(p_{01},\omega_{\beta}^*(m))+1}\Tc^{L(p_{01},\omega_{\beta}^*(m))}(p_{01})},
& \mbox{if}~ p_{01}\le\omega^*_{\beta}(m)<\omega_o\\
\frac{1}{1-\beta}, & \mbox{if}~\omega^*_{\beta}(m)\ge\omega_o
\end{array} \right..~~~\label{eq:PcloseformDp01}\\
D_{\beta,m}(p_{11})&=&\left\{\begin{array}{ll}
\frac{\beta(1-p_{11})D_{\beta,m}(p_{01})}{1-\beta p_{11}}, &
\mbox{if}~\omega^*_{\beta}(m)<p_{11}\\
\frac{1}{1-\beta}, &\mbox{if}~ \omega^*_{\beta}(m)\ge p_{11}
\end{array} \right.,\label{eq:PcloseformDp11}
\end{eqnarray}

\begin{itemize}
\item {\em Case 2: Negatively correlated channel ($p_{11}< p_{01}$)}
\end{itemize}
\begin{eqnarray}
D_{\beta,m}(p_{11})&=&\left\{\begin{array}{ll} 0, &
\mbox{if}~\omega^*_{\beta}(m)<p_{11}\\
\frac{1-\beta(1-p_{01})}
{1-\beta(1-p_{01})-\beta^2\Tc(p_{11})(1-\beta)-\beta^3p_{01}}, &
\mbox{if}~
p_{11}\le\omega^*_{\beta}(m)<\Tc(p_{11})\\
\frac{1}{1-\beta}, & \mbox{if}~ \omega^*_{\beta}(m)\ge\Tc(p_{11})
\end{array} \right..\label{eq:NcloseformDp11}\\
D_{\beta,m}(p_{01})&=&\left\{\begin{array}{ll} \frac{\beta
p_{01}D_{\beta,m}(p_{11})}{1-\beta(1-p_{01})}, &
\mbox{if}~ \omega^*_{\beta}(m)<p_{01}\\
\frac{1}{1-\beta}, & \mbox{if}~ \omega^*_{\beta}(m)\ge p_{01}
\end{array} \right.,\label{eq:NcloseformDp01}
\end{eqnarray}
\end{lemma}

\begin{proof}
The process of solving for $D_{\beta,m}(p_{01})$ and
$D_{\beta,m}(p_{11})$ is similar to that of solving for
$V_{\beta,m}(p_{01})$ and $V_{\beta,m}(p_{11})$. Details are
omitted. $D_{\beta,m}(p_{01})$ and $D_{\beta,m}(p_{11})$ can also be
obtained by taking the derivatives of $V_{\beta,m}(p_{01})$ and
$V_{\beta,m}(p_{11})$ with respect to $m$.
\end{proof}

\vspace{0.5em}

We point out that $V_{\beta,m}(\omega)$ is not differentiable in $m$
at every point (\ie the left derivative may not equal to the right
derivative). Suppose that $V_{\beta,m}(\omega)$ is not
differentiable at $m_0$. Then it can be shown that the left
derivative at $m_0$ corresponds to the case when the threshold
$\omega^*_{\beta}(m_0)$ is included in the active set while the
right derivative corresponds to the case when
$\omega^*_{\beta}(m_0)$ is included in the passive set. In this
paper, we include the threshold in the passive set
(see~\eqref{eqn:passiveset1}), \ie we choose the passive action when
both actions are optimal. As a consequence, we consider the right
derivative of $V_{\beta,m}(\omega)$ when it is not differentiable.

The following lemma shows the piecewise constant (a stair function) and monotonically increasing
properties of $D_{\beta,m}(\omega)$ as a function of
$m$. These properties allow us to develop an efficient algorithm for computing
a performance upper bound as shown in Sec.~\ref{subsec:performancediscount}.

\begin{lemma}\label{lemma:passivetimeform}
The total discounted passive time $D_{\beta,m}(\omega)$ as a function of $m$ is monotonically increasing
and piecewise constant (with countable pieces for $p_{11}\ge p_{01}$ and finite pieces
for $p_{11}< p_{01}$). Equivalently, the value function
$V_{\beta,m}(\omega)$ is piecewise linear and convex in $m$.
\end{lemma}
\begin{proof}
The piecewise constant property follows directly from~\eqref{eq:dcloseform}
and Lemma~\ref{lemma:closeformD} and is illustrated in Fig.~\ref{fig:ND} and
Fig.~\ref{fig:PD}.
The monotonicity of $D_{\beta,m}(\omega)$ applies to a general restless bandit and has been stated
without proof by Whittle~\cite{whittle}. We provide a proof below
for completeness.

We show that $V_{\beta,m}(\omega)$ is convex in $m$, \ie for any
$0\le\alpha\le1, m_1,m_2\in\Rc$,
\begin{eqnarray}\label{eq:convexV}
\alpha V_{\beta,m_1}(\omega)+(1-\alpha)V_{\beta,m_2}(\omega)\ge
V_{\beta,\alpha m_1+(1-\alpha)m_2}(\omega).
\end{eqnarray}
Consider the optimal policy $\pi$ under subsidy $\alpha
m_1+(1-\alpha)m_2$. If we apply $\pi$ to the system with subsidy
$m_1$, the total discounted reward will be
\[V_{\beta,\alpha m_1+(1-\alpha)m_2}(\omega)+D_{\beta,\alpha
m_1+(1-\alpha)m_2}(\omega)((1-\alpha)(m_1-m_2)).\] Since $\pi$ may
not be the optimal policy under subsidy $m_1$, we have
\begin{eqnarray}
V_{\beta,m_1}(\omega)\ge V_{\beta,\alpha
m_1+(1-\alpha)m_2}(\omega)+D_{\beta,\alpha
m_1+(1-\alpha)m_2}(\omega)((1-\alpha)(m_1-m_2)).\label{eqn:ddd1}
\end{eqnarray}
Similarly,
\begin{eqnarray}
V_{\beta,m_2}(\omega)\ge V_{\beta,\alpha
m_1+(1-\alpha)m_2}(\omega)+D_{\beta,\alpha
m_1+(1-\alpha)m_2}(\omega)(\alpha(m_2-m_1)).\label{eqn:ddd2}
\end{eqnarray}
\eqref{eq:convexV} thus follows from \eqref{eqn:ddd1} and
\eqref{eqn:ddd2}.
\end{proof}

\subsection{Indexability and Whittle's Index Policy}\label{subsec:indexabilitydiscount}

With the threshold structure of the optimal policy and the
closed-form expressions of the value function and discounted passive
time, we are ready to establish the indexability and solve for Whittle's index.

\begin{theorem}\label{thm:indexability} The restless multi-armed bandit
process $(\Omega(1), \{{\bf P}_i\}_{i=1}^N, \{B_i\}_{i=1}^N, \beta)$
is indexable.
\end{theorem}

\begin{proof}
The proof is based on Lemma~\ref{lemma:thresholdPolicy} and Lemma~\ref{lemma:closeformD}. Details are
given in Appendix B.
\end{proof}

\vspace{1em}

\begin{theorem}\label{thm:whittleindex}

Whittle's index $W_{\beta}(\omega)\in\mathbb{R}$ for arm $i$ of the
RMBP $(\Omega(1), \{{\bf P}_i\}_{i=1}^N, \{B_i\}_{i=1}^N, \beta)$ is
given as follows.
\begin{itemize}
\item \emph{Case 1: Positively correlated channel ($p^{(i)}_{11}\ge p^{(i)}_{01}$).}
\end{itemize}
\begin{eqnarray}
W_{\beta}(\omega)=\left\{\begin{array}{ll} \omega B_i, &
\mbox{if}~\omega\le p^{(i)}_{01}~\mbox{or}
~\omega\ge p^{(i)}_{11}\\
\frac{\omega}{1-\beta p^{(i)}_{11}+\beta\omega}B_i, &
\mbox{if}~\omega^{(i)}_o\le
\omega<p^{(i)}_{11}\\
\frac{\omega-\beta\Tc^1(\omega)+C_2(1-\beta)(\beta(1-\beta
p^{(i)}_{11})-\beta(\omega-\beta\Tc^1(\omega)))} {1-\beta
p^{(i)}_{11}-C_1(\beta(1-\beta
p^{(i)}_{11})-\beta(\omega-\beta\Tc^1(\omega)))}B_i,
 & \mbox{if}~p^{(i)}_{01}<
\omega<\omega^{(i)}_o
\end{array}\right.,\label{eqn:whittleindexcase1}
\end{eqnarray}
where $C_1=\frac{(1-\beta
p^{(i)}_{11})(1-\beta^{L(p^{(i)}_{01},\omega)})} {(1-\beta
p^{(i)}_{11})(1-\beta^{L(p^{(i)}_{01},\omega)+1})+(1-\beta)\beta^{L(p^{(i)}_{01},\omega)+1}
\Tc^{L(p^{(i)}_{01},\omega)}(p^{(i)}_{01})}$,\\[0.5em]
$~~~~~~~~~C_2=\frac{\beta^{L(p^{(i)}_{01},\omega)}\Tc^{L(p^{(i)}_{01},\omega)}(p^{(i)}_{01})}
{(1-\beta
p^{(i)}_{11})(1-\beta^{L(p^{(i)}_{01},\omega)+1})+(1-\beta)\beta^{L(p^{(i)}_{01},\omega)+1}\Tc^{L(p^{(i)}_{01},\omega)}(p^{(i)}_{01})}$.\\
\begin{itemize}

\item \emph{Case 2: Negatively correlated channel ($p^{(i)}_{11}<p^{(i)}_{01}$).}
\end{itemize}
\begin{eqnarray}
W_{\beta}(\omega)=\left\{\begin{array}{ll} \omega B_i, &
\mbox{if}~\omega\le p^{(i)}_{11}~\mbox{or}
~\omega\ge p^{(i)}_{01}\\
\frac{\beta
p^{(i)}_{01}+\omega(1-\beta)}{1+\beta(p^{(i)}_{01}-\omega)}B_i, &
\mbox{if}~\Tc^1(p^{(i)}_{11})\le
\omega<p^{(i)}_{01}\\
\frac{(1-\beta+\beta C_4)(\beta
p^{(i)}_{01}+\omega(1-\beta))}{1-\beta(1-p^{(i)}_{01})-C_3(\beta^2p^{(i)}_{01}+\beta\omega-\beta^2\omega)}B_i,
 & \mbox{if}~\omega^{(i)}_o\le\omega<\Tc^1(p^{(i)}_{11})\\[1em]
\frac{(1-\beta)(\beta
p^{(i)}_{01}+\omega-\beta\Tc^1(\omega))-C_4\beta(\beta\Tc^1(\omega)-\beta
p^{(i)}_{01}-\omega)}
{1-\beta(1-p^{(i)}_{01})+C_3\beta(\beta\Tc^1(\omega)-\beta
p^{(i)}_{01}-\omega)}B_i,
 & \mbox{if}~p^{(i)}_{11}<\omega<\omega^{(i)}_o\\
\end{array}\right.,\label{eqn:whittleindexcase2}
\end{eqnarray}
where $C_3=\frac{1-\beta(1-p^{(i)}_{01})}{1+(1+\beta)\beta
p^{(i)}_{01}-\beta^2\Tc^1(p^{(i)}_{11})}$ and
$C_4=\frac{\beta\Tc^1(p^{(i)}_{11})(1-\beta)+\beta^2p^{(i)}_{01}}{1+(1+\beta)\beta
p^{(i)}_{01}-\beta^2\Tc^1(p^{(i)}_{11})}$.
\end{theorem}

\vspace{1em}

\begin{proof} By the definition of Whittle's index, for a given belief state $\omega$,
its Whittle's index is the subsidy $m$ that is the solution to the
following equation of $m$:
\begin{eqnarray}\underbrace{\omega+\beta(\omega
V_{\beta,m}(p_{11})+(1-\omega)V_{\beta,m}(p_{01}))}_{V_{\beta,m}(\omega;u=1)}=\underbrace{m+\beta
V_{\beta,m}(\Tc^1(\omega))}_{V_{\beta,m}(\omega;u=0)}.\label{eqn:whittleindex}
\end{eqnarray}
From the closed-form expressions for $V_{\beta,m}(p_{11})$,
$V_{\beta,m}(p_{01})$ and $V_{\beta,m}(\Tc^1(\omega))$ given in
Lemma~\ref{lemma:valuefunction}, we can
solve~\eqref{eqn:whittleindex} and obtain Whittle's index.
\end{proof}

\vspace{1em}

The following properties of Whittle's index $W_{\beta}(\omega)$
follow from Theorem~\ref{thm:indexability} and
Theorem~\ref{thm:whittleindex}.

\begin{corollary}\label{cor:Wproperty}
{\it Properties of Whittle's Index}
\begin{itemize}
\item $W_\beta(\omega)$ is a monotonically increasing function of
$\omega$. As a consequence, Whittle's index policy is equivalent
to the myopic policy for stochastically identical arms.
\item For a positively correlated channel ($p_{11}\ge p_{01}$),
$W_\beta(\omega)$ is piecewise concave with countable pieces. More
specifically, $W_{\beta}(\omega)$ is linear in $[0,p_{01}]$ and
$[p_{11},1]$, concave in $[\omega_o,p_{11})$, and piecewise concave
with countable pieces in $(p_{01},\omega_0)$ (see
Fig.~\ref{fig:plotindex}-left).
\item For a negatively correlated channel ($p_{11}<p_{01}$),
$W_\beta(\omega)$ is piecewise convex with finite pieces. More
specifically, $W_{\beta}(\omega)$ is linear in $[0,p_{11}]$ and
$[p_{01},1]$, concave in $(p_{11},\omega_o)$,
$[\omega_o,\Tc(p_{11}))$, and $[\Tc(p_{11}),p_{01})$ (see
Fig.~\ref{fig:plotindex}-right).
\end{itemize}
\end{corollary}

\vspace{1em}


The equivalency between Whittle's index policy and the myopic policy
is particularly important. It allows us to establish the structure and
optimality of Whittle's index policy by examining the myopic policy which
has a very simple index form.

Note that the region of $[p_{01},\omega_o)$ for a positively
correlated arm is the most complex. The infinite but countable
concave pieces of Whittle's index in this region correspond to each
possible value of the crossing time
$L(p_{01},\omega)\in\{1,2,\cdots\}$. This region presents most of
the difficulties in analyzing the performance of Whittle's index
policy as shown in the next subsection.

\begin{figure}[h]
\begin{minipage}{4in}
\hspace{-3em}\centerline{
\begin{psfrags}
\scalefig{1}\epsfbox{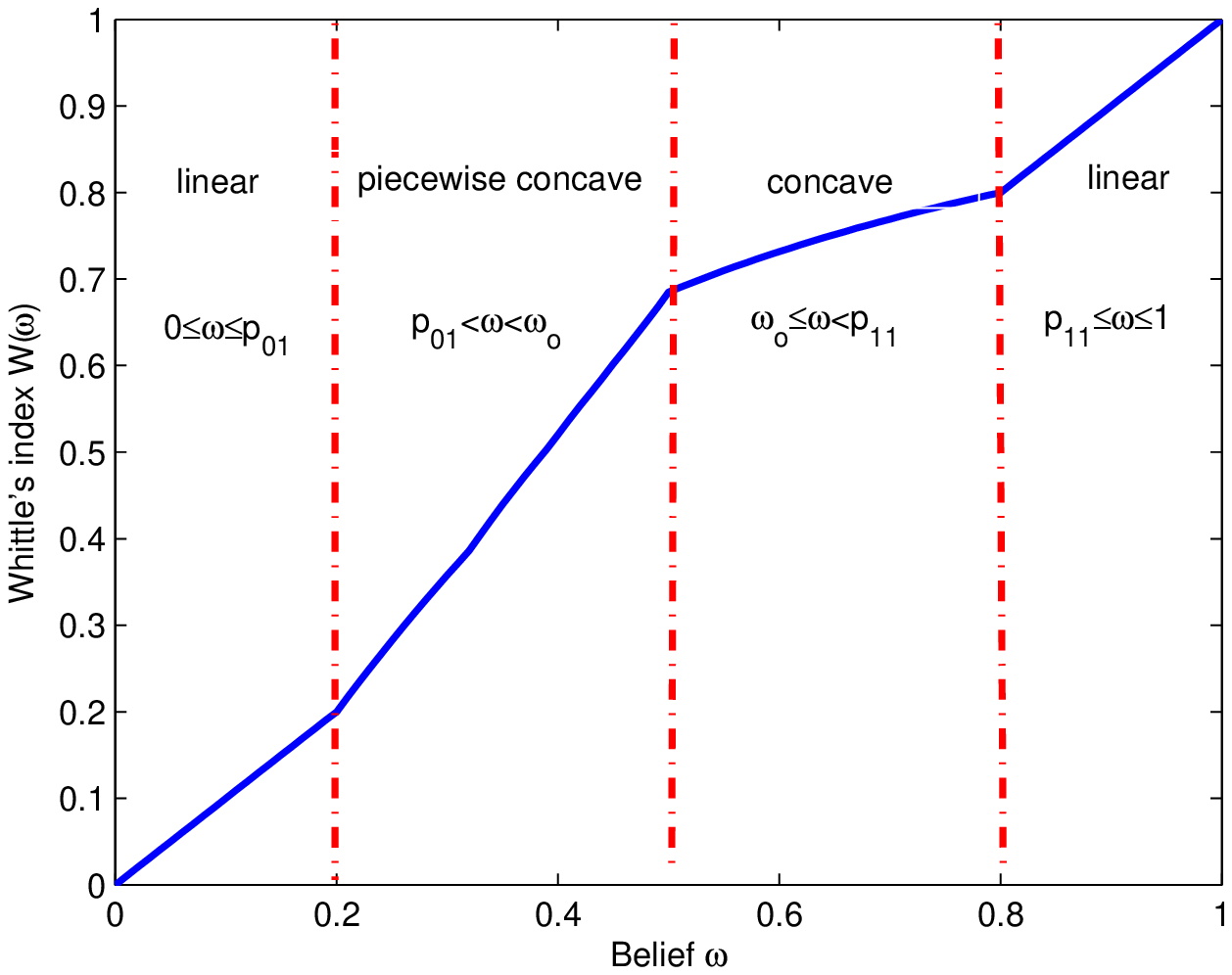}
\end{psfrags}}
\end{minipage}
\begin{minipage}{4in}
\hspace{-4em}\centerline{
\begin{psfrags}
\scalefig{1}\epsfbox{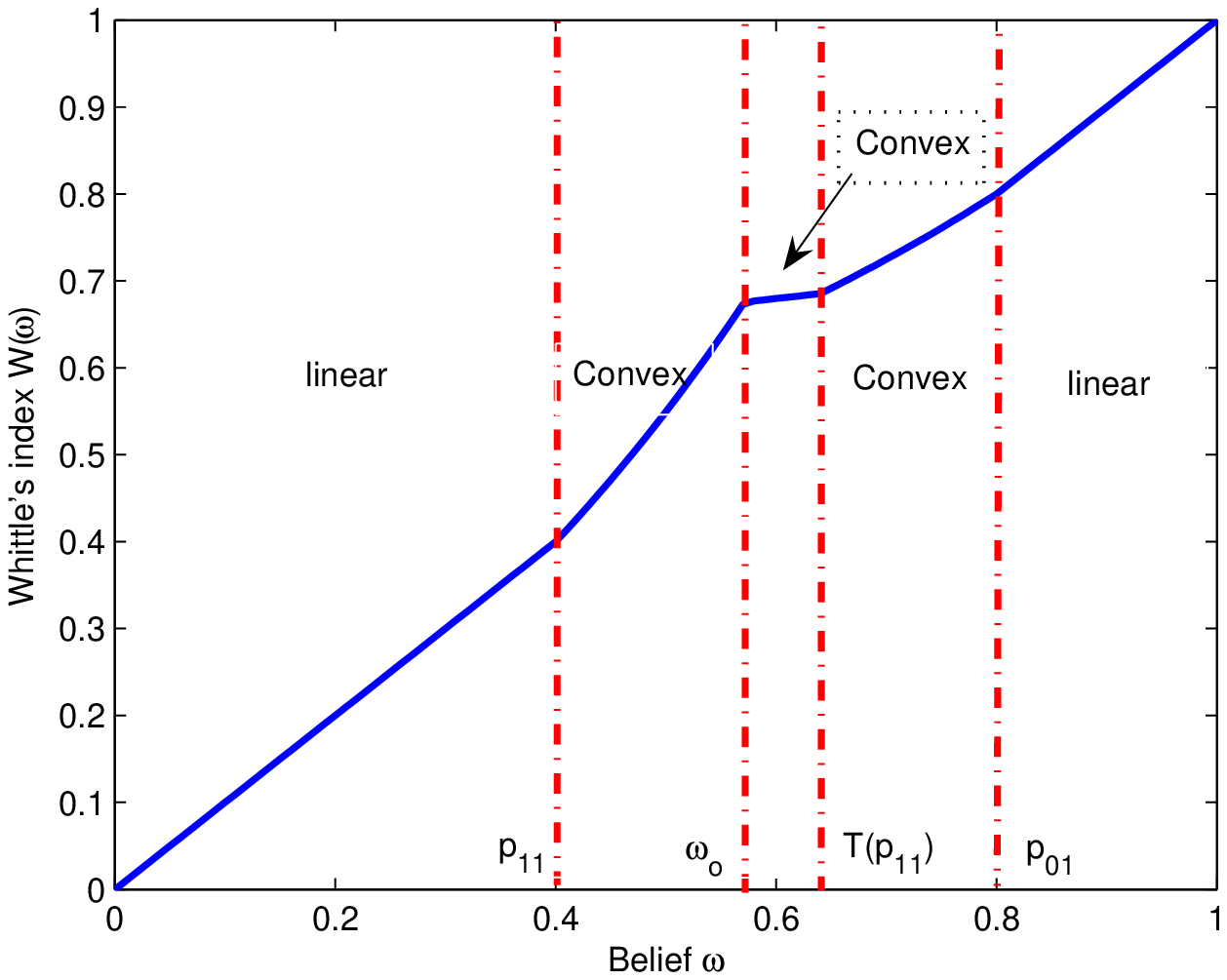}
\end{psfrags}}
\end{minipage}
\caption{Whittle's index (left: $p_{11}=0.8,~p_{01}=0.2,~\beta=0.9$;
right: $p_{11}=0.4,~p_{01}=0.8,~\beta=0.9$).} \label{fig:plotindex}
\end{figure}

\subsection{Performance of Whittle's Index
Policy}\label{subsec:performancediscount}

\subsubsection{The optimality of Whittle's Index Policy under a Relaxed Constraint}

Whittle's index policy is the optimal solution to a Lagrangian
relaxation of RMBP~\cite{whittle}. Specifically, the number of
activated arms can vary over time provided that its discounted
average over the infinite horizon equals to $K$. Let $K(t)$ denote
the number of arms activated in slot $t$. The relaxed constraint is
given by
\begin{eqnarray}\label{eq:relaxCon}
\mathbb{E}_{\pi}[(1-\beta)\Sigma_{t=1}^{\infty}\beta^{t-1}K(t)]=K.
\end{eqnarray}
Let $\bar{V}_{\beta}(\Omega(1))$ denote the maximum expected total
discounted reward that can be obtained under this relaxed constraint
when the initial belief vector is $\Omega(1)$. Based on the
Lagrangian multiplier theorem, we have~\cite{whittle}
\begin{eqnarray}\label{eqn:relax}
\bar{V}_{\beta}(\Omega(1))=
\inf_{m}\{\Sigma_{i=1}^NV_{\beta,m}^{(i)}(\omega_i(1))-m\frac{(N-K)}{1-\beta}\},
\end{eqnarray}
where $V_{\beta,m}^{(i)}(\omega)$ is the value function of the
single-armed bandit process with subsidy $m$ that corresponds to
the $i$-th channel.

The above equation reveals the role of the subsidy $m$ as the
Lagrangian multiplier and the optimality of Whittle's index policy
for RMBP under the relaxed constraint given in~\eqref{eq:relaxCon}.
Specifically, under the relaxed constraint, Whittle's index policy
is implemented by activating, in each slot, those arms whose current
states have a Whittle's index greater than a constant $m^*$. This
constant $m^*$ is the Lagrangian multiplier that makes the relaxed
constraint given in~\eqref{eq:relaxCon} satisfied, or equivalently,
the Lagrangian multiplier that achieves the infimum
in~\eqref{eqn:relax}. It is not difficult to see that Whittle's
index policy implemented by comparing to a constant $m^*$ is the
optimal policy (\ie achieves $\bar{V}_{\beta}(\Omega(1))$) for RMBP
under the relaxed constraint.

\subsubsection{An Upper Bound of The Optimal Performance}

Under the strict constraint of $K(t)=K$ for all $t$, Whittle's index
policy is implemented by activating those $K$ arms with the largest
indices in each slot. Its optimality is lost in general.

\begin{figure}[h]
\centerline{
\begin{psfrags}
\psfrag{g}[c]{$G_{\beta}(\Omega(1),m)$
}\psfrag{v}[c]{$\bar{V}_{\beta}(\Omega(1))$}
\psfrag{m1}[c]{$m^*$}\psfrag{m}[c]{$m$}
\scalefig{0.6}\epsfbox{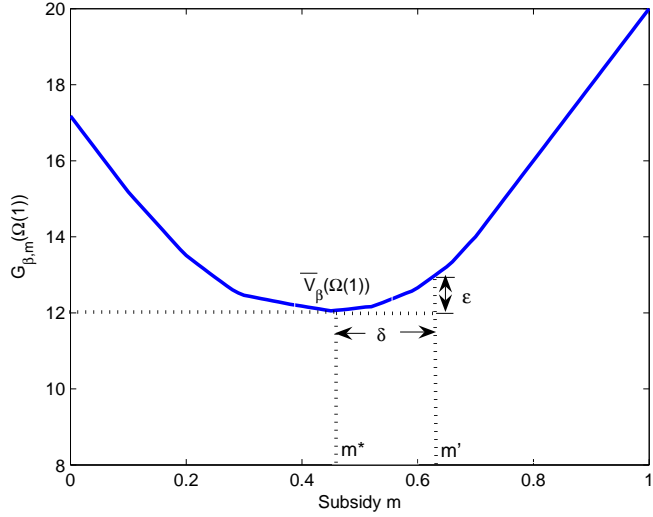}
\end{psfrags}}
\caption{The optimal performance under relaxed constraint ($N
=8,~M=4,~
\{p^{(i)}_{01}\}_{i=1}^8=[0.2,0.5,0.8,0.1,0.6,0.2,0.3,0.8],~
\{p^{(i)}_{11}\}_{i=1}^8=[0.4,0.1,0.3,0.6,0.2,0.8,0.7,0.6],~   
B_i=1~\mbox{for all}~i=1,\hdots,8,~ \gamma=0.8$).}
\label{fig:perform}
\end{figure}

Let $V_{\beta}(\Omega(1))$ denote the maximum expected total
discounted reward of the RMBP under the strict constraint that $K(t)=K$ for all $t$.
It is obvious that
\[V_{\beta}(\Omega(1))\le\bar{V}_{\beta}(\Omega(1)).\]
$\bar{V}_{\beta}(\Omega(1))$ thus provides a performance benchmark
for all RMBP policies, including Whittle's index policy.
Unfortunately, $\bar{V}_{\beta}(\Omega(1)$ as given in
\eqref{eqn:relax} is, in general, difficult to obtain due to the
complexity of calculating the value functions of all arms and
searching for the infimum over an uncountable space. For the problem
at hand, however, we have obtained $V_{\beta,m}^{(i)}(\omega_i(1))$
in closed-form as given in Lemma~\ref{lemma:valuefunction}.
Furthermore, the piecewise constant structure of the discounted
passive time $D_{\beta,m}^{(i)}(\omega_i(1))$ given in
Lemma~\ref{lemma:passivetimeform} leads to efficient algorithms for
searching for the infimum of the value functions over $m$ as shown
below.

Let
\[G_{\beta,m}(\Omega(1))=\Sigma_{i=1}^NV_{\beta,m}^{(i)}(\omega_i(1))-m\frac{(N-K)}{1-\beta}.\]
We then have $\bar{V}_{\beta}(\Omega(1))=\inf_m
G_{\beta,m}(\Omega(1),m)$. From Lemma~\ref{lemma:passivetimeform},
it is easy to see that $G_{\beta,m}(\Omega(1))$ is convex in $m$
as illustrated in Fig.~\ref{fig:perform}. The infimum of
$G_{\beta,m}(\Omega(1))$ is achieved at $m^*$ at which the
derivative of $G_{\beta,m}(\Omega(1))$ with respect to $m$ becomes
nonnegative for the first time (note that $G_{\beta,m}(\Omega(1))$
is not differentiable at every $m$, and we
consider the right derivative when it is not differentiable).
Equivalently,
\[
m^*=\sup\{m:
\frac{d(G_{\beta,m}(\Omega(1)))}{dm}=\Sigma_{i=1}^ND_{\beta,m}^{(i)}(\omega_i(1))-\frac{(N-K)}{1-\beta}\le
0\}.
\]

\begin{figure}[h]
\centerline{
\begin{psfrags}
\psfrag{w1}[c]{\footnotesize$W_{\beta}(\omega(1))$}
\psfrag{p1}[c]{\footnotesize$W_{\beta}(p_{11})$}
\psfrag{p0}[c]{\footnotesize$W_{\beta}(p_{01})$}\psfrag{tw1}[c]{\footnotesize
~~$W_{\beta}(\Tc^1(\omega(1)))$} \psfrag{tp1}[c]{\footnotesize
$W_{\beta}(\Tc^1(p_{11}))$} \psfrag{0}[c]{$0$}\psfrag{1}[c]{$1$}
\psfrag{D}[c]{$D_{\beta,m}(\omega(1))$} \psfrag{m}[c]{$m$}
\scalefig{0.9}\epsfbox{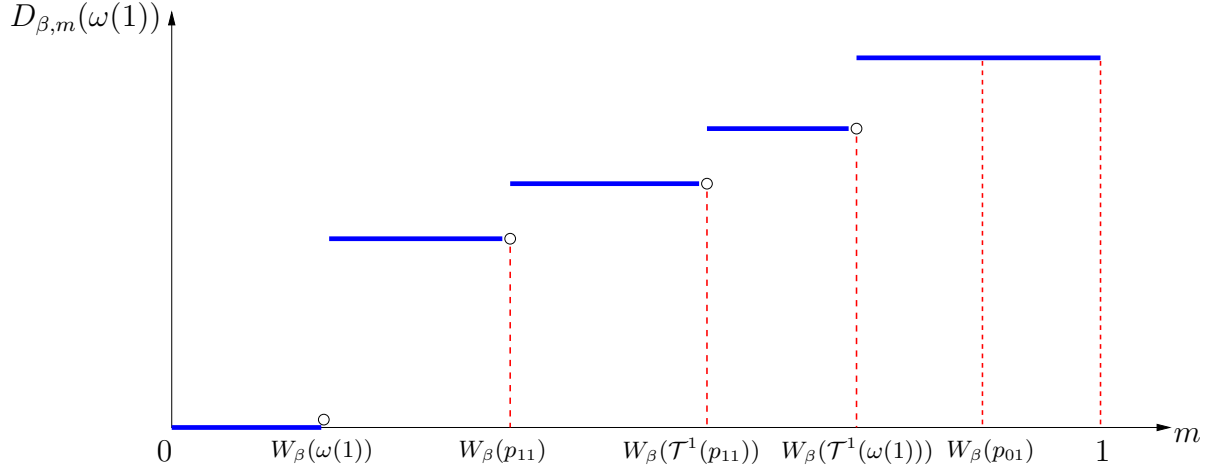}
\end{psfrags}}
\caption{The passive time for different regions ($p_{11}<p_{01}$).}
\label{fig:ND}
\end{figure}

\begin{figure}[h]
\centerline{
\begin{psfrags}
\psfrag{ga}[c]{Gray Area (infinite pieces)}
\psfrag{p11}[c]{~~~~\footnotesize $W_{\beta}(p_{11})$}
\psfrag{barw}[l]{\footnotesize$W_{\beta}(\bar{\omega})$}
\psfrag{wo}[c]{\footnotesize$W_{\beta}(\omega_o)$}
\psfrag{w1}[c]{\footnotesize$W_{\beta}(\omega(1))$}
\psfrag{p1}[c]{\footnotesize$W_{\beta}(p_{11})$}
\psfrag{p0}[c]{\footnotesize$W_{\beta}(p_{01})$}\psfrag{tp0}[c]{\footnotesize
~~$W_{\beta}(\Tc^1(p_{01}))$} \psfrag{tlp0}[c]{\footnotesize
~~~$W_{\beta}(\Tc^{L(p_{01},\bar{\omega})-1}(p_{01}))$}
\psfrag{0}[c]{$0$}\psfrag{1}[c]{$1$}
\psfrag{D}[c]{$D_{\beta,m}(\omega(1))$~~~~} \psfrag{m}[c]{~~$m$}
\scalefig{0.9}\epsfbox{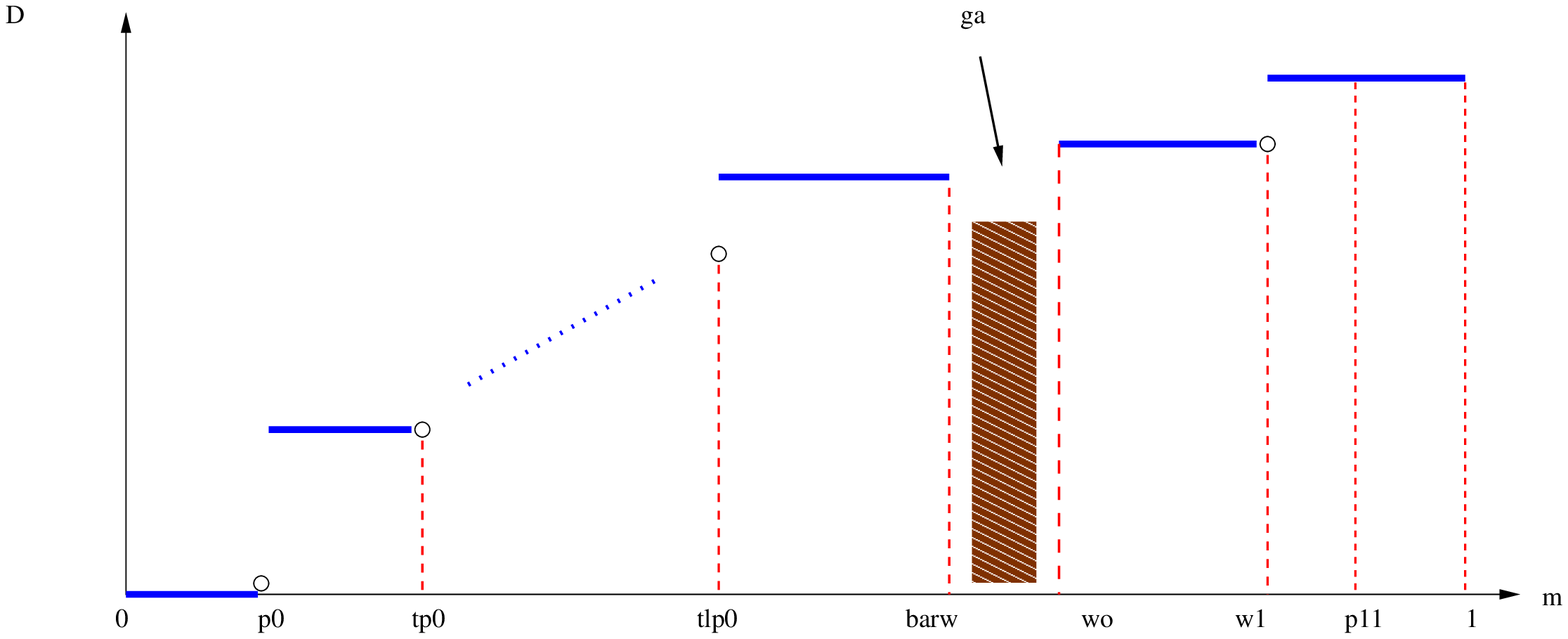}
\end{psfrags}}
\caption{The passive time for different regions ($p_{11}\ge
p_{01}$).} \label{fig:PD}
\end{figure}

From Lemma~\ref{lemma:passivetimeform},
$D_{\beta,m}^{(i)}(\omega_i(1))$ is piecewise constant for each
channel (see Fig.~\ref{fig:ND} and Fig.~\ref{fig:PD}). We can thus
partition the range of $m$ into disjoint regions such that
$\frac{d(G_{\beta,m}(\Omega(1)))}{dm}$ is constant in each region.
To obtain $m^*$, we only need to check each region successively
until $\frac{d(G_{\beta,m}(\Omega(1)))}{dm}$ becomes nonnegative for
the first time (due to the monotonically increasing property of $D_{\beta,m}^{(i)}(\omega_i(1))$
in $m$). The difficulty is that for a positively correlated
channel, there are infinite constant regions of
$D_{\beta,m}^{(i)}(\omega_i(1))$ (see
Fig.~\ref{fig:PD}). However, we can find an arbitrarily small interval
$(W_{\beta}(\bar{\omega}), W_{\beta}(\omega)]$---referred to as the gray area---outside
which there are only finite number of constant regions of $D_{\beta,m}^{(i)}(\omega_i(1))$.
By setting the gray area for each positively correlated channel small
enough, we can find an $m'$ that is arbitrarily close to $m^*$ so that
$G_{\beta,m'}(\Omega(1)))-G_{\beta,m^*}(\Omega(1))\le\epsilon$ for
any $\epsilon>0$. Specifically, we set the length of the gray area
for each positively correlated channel to $\frac{\delta}{N}$ (\ie
$W_{\beta}(\omega_o)-W_{\beta}(\bar{\omega})\le\frac{\delta}{N}$)
where $\delta=\frac{\epsilon(1-\beta)}{K}$. The total length of the
gray area over all channels is thus at most $\delta$, \ie
$m'-m^*\le\delta$. Based on the convexity of
$G_{\beta,m}(\Omega(1))$, the maximum derivative of
$G_{\beta,m}(\Omega(1))$ for $m^*\le m\le1$ is achieved at $m=1$,
which is equal to $\frac{K}{1-\beta}$. Thus, we have
\[G_{\beta,m'}(\Omega(1)))-G_{\beta,m^*}(\Omega(1))\le\frac{K}{1-\beta}(m'-m^*)
\le\frac{\delta K}{1-\beta}=\epsilon.\]

We point out that if $m^*$ does
not fall into the gray area, the algorithm will obtain $m^*$ and
$\bar{V}_{\beta}(\Omega(1))$ without error. In the special case when
every channel is negatively correlated, the algorithm will always output the exact value of $m^*$
and $\bar{V}_{\beta}(\Omega(1))$. The detailed algorithm is given in
Fig.~\ref{fig:upperalg}. The complexity of this algorithm is given
in the following theorem.

\begin{figure}[htbp]
\begin{center}
\noindent\fbox{
\parbox{6.5in}
{ \centerline{\underline{{\bf Computing the Performance Upper Bound within $\epsilon$-Accuracy}}}
Input an $\epsilon>0$. Set
$\delta=\frac{\epsilon(1-\beta)}{K}$ and $j=0$.
\begin{enumerate}
\item For each negatively correlated channel $i$, calculate $W_{\beta}(p_{11}^{(i)})$,
$W_{\beta}(p_{01}^{(i)})$, and $W_{\beta}(\Tc(p_{11}^{(i)}))$. If
$\omega_i(1)<\omega_o^{(i)}$, calculate
$W_{\beta}(\omega_i(1))$ and $W_{\beta}(\Tc^1(\omega_i(1)))$;
otherwise only calculate $W_{\beta}(\omega_i(1))$.
\item For each positively correlated channel $i$, calculate
$W_{\beta}(p_{01}^{(i)})$, $W_{\beta}(p_{11}^{(i)})$, and
$W_{\beta}(\omega_o^{(i)})$. Search for an
$\bar{\omega}^{(i)}\in[\omega_o^{(i)}-\frac{\delta}{N},\omega_o^{(i)})$
such that $W_{\beta}(\omega_o^{(i)})\ge
W_{\beta}(\bar{\omega}^{(i)})-\frac{\delta}{N}$. Let $l_i$ be the
smallest integer such that
$\Tc^{l_i}(p_{01}^{(i)})>\bar{\omega}^{(i)}$. Calculate
$W_{\beta}(\Tc^k(p_{01}^{(i)}))$ for all $1\le k\le l_i$. If
$\omega_i(1)<\omega_o^{(i)}$, then let $d_i$ be the smallest integer
such that $\Tc^{d_i}(\omega_i(1))>\bar{\omega}^{(i)}$ and calculate
$W_{\beta}(\Tc^k(\omega_i(1)))$ for all $1\le k\le d_i$; otherwise
only calculate $W_{\beta}(\omega_i(1))$. Set the gray area
$V=\cup_{i}[\min\{W_{\beta}(\Tc^{l_i}(p_{01}^{(i)})),W_{\beta}(\Tc^{d_i}(\omega_i(1)))\},W_{\beta}(\omega_o^{(i)}))$.
\item  Order all Whittle's indices calculated in Step 1 and 2 by
the ascending order. Let $[a_1,...a_h]$ denote the ordered Whittle's indices.
Set $a_0=0$ and $a_{h+1}=1$.
\item If $[a_j,a_{j+1})\nsubseteq V$, calculate
$D=\Sigma_{k=1}^N D_{\beta,m}^{(k)}(\omega_k(1))-\frac{(N-K)}{1-\beta}$
for $m\in [a_j,a_{j+1})$ according
to~\eqref{eq:dcloseform} (note that every $D_{\beta,m}^{(k)}(\omega_k(1))$ is constant
for $m\in [a_j,a_{j+1})$). If $D$ is nonnegative, go to Step 5;
otherwise set $j=j+1$ and repeat Step 4.
\item Calculate $G=G_{\beta,m}(\Omega(1))$ when $m\in [a_j,a_{j+1})$
according to~\eqref{eq:vcloseform}. Output $m'=a_j$ and $G$.
\end{enumerate}
}} \caption{Algorithm for computing the upper bound of the optimal
performance.}\label{fig:upperalg}
\end{center}
\end{figure}
\begin{theorem}\label{thm:algorithm}
For any $\epsilon>0$, the algorithm given in~Fig.~\ref{fig:upperalg}
runs in at most $O(N^2\log N)$ time to output a value $G$ that is
within $\epsilon$ of $\bar{V}_{\beta}(\Omega(1))$ for any $\epsilon>0$.
\end{theorem}
\begin{proof}
See Appendix C.
\end{proof}

To find the infimum of $G_{\beta}(\Omega(1),m)$, we can also carry out a
binary search on subsidy $m$. 
It can
be shown that this algorithm runs in $O(N(\log N)^2)$ time. However,
it cannot output the exact value of $m^*$ and
$\bar{V}_{\beta}(\Omega(1))$.

Fig.~\ref{fig:relaxpolicy} shows an example of the performance
of Whittle's index policy. It demonstrates the near optimal performance
of Whittle's index policy and the tightness of the performance upper bound.

\begin{figure}[htb]
\centerline{
\begin{psfrags}
\scalefig{0.6}\epsfbox{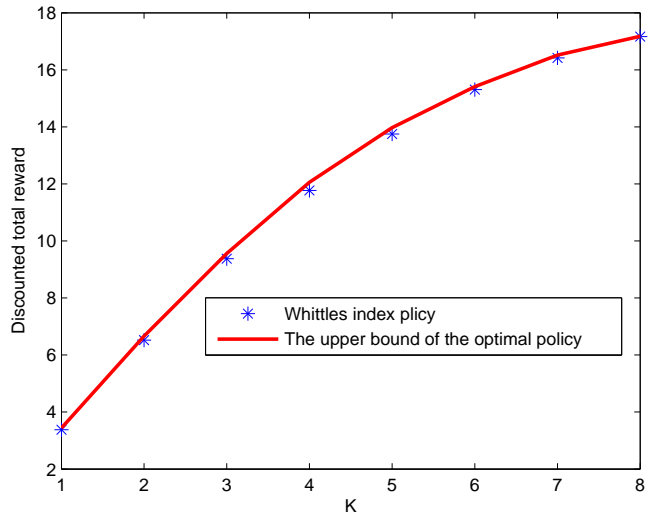}
\end{psfrags}}
\caption{The Performance of Whittle's index policy ($N=8,~
\{p^{(i)}_{01}\}_{i=1}^8= \{0.2,0.5,0.8,0.1,0.6,0.2,0.3,
0.8\},~\{p^{(i)}_{11}\}_{i=1}^8 =\{0.4,0.1,0.3,0.6,0.2,0.8,0.7,
0.6\},~B_i=1~\mbox{for}~i=1,\hdots,8,~\mbox{and}~\beta=0.8$).}
\label{fig:relaxpolicy}
\end{figure}

\section{Whittle's index under Average reward criterion}\label{sec:whittleindexpolicyavg}

In this section, we investigate Whittle's index policy under the
average reward criterion and establish results parallel to those
obtained under the discounted reward criterion in
Sec.~\ref{sec:discount}.

%
%
%
%

\subsection{The Value Function and The Optimal
Policy}\label{subsec:vpavg}

First, we present a general result by Dutta~\cite{Dutta} on the
relationship between the value function and the optimal policy under
the total discounted reward criterion and those under the average
reward criterion. This result allows us to study Whittle's index
policy under the average reward criterion by examining its limiting
behavior as the discount factor $\beta\rightarrow1$.

{\it Dutta's Theorem}~\cite{Dutta}.
Let $\Fc$ be the belief space of a POMDP and $V_{\beta}(\Omega)$ the
value function with discount factor $\beta$ for belief
$\Omega\in\Fc$. The POMDP satisfies the value boundedness condition
if there exist a belief $\Omega'$, a real-valued function
$c_1(\Omega):\Fc\rightarrow\Rc$, and a constant $c_2<\infty$ such
that
\[c_1(\Omega)\le V_{\beta}(\Omega)-V_{\beta}(\Omega')\le
c_2,\] for any $\Omega\in\Fc$ and $\beta\in[0,1)$. Under the
value-boundedness condition, if a series of optimal policies
$\pi_{\beta_k}$ for a POMDP with discount factor $\beta_k$ pointwise
converges to a limit $\pi^*$ as $\beta_k\rightarrow1$, then $\pi^*$
is the optimal policy for the POMDP under the average reward
criterion. Furthermore, let $J(\Omega)$ denote the maximum expected
average reward over the infinite horizon starting from the initial
belief $\Omega$. We have
\[J(\Omega)=\lim_{\beta_k\rightarrow
1}(1-\beta_k)V_{\beta_k}(\Omega)\] and $J(\Omega)=J$ is independent
of the initial belief $\Omega$.

Next, we will show that the single-armed bandit process with subsidy $m$
under the discounted reward criterion (see
Sec.~\ref{subsec:whittleindexpolicy}) satisfies the valueboundedness
condition.


\begin{lemma}\label{lemma:valueboundedness}
The single-armed bandit process with subsidy under the discounted
reward criterion satisfies the value-boundedness condition. More
specifically, we have\footnote{Here we do not consider the trivial
case that the arm has absorbing states.}
\begin{eqnarray}
|V_{\beta,m}(\omega)-V_{\beta,m}(\omega')|\le c+1,~~\mbox{for
all}~\omega,\omega'\in[0,1],
\end{eqnarray}
where $c=\max\{\frac{2}{1-p_{11}}, \frac{2}{p_{01}}\}$.
\end{lemma}

\begin{proof}
See Appendix D.
\end{proof}


Under the value boundedness condition, the optimal policy for the
single-armed bandit process with subsidy under the average reward
criterion can be obtained from the limit of any pointwise convergent
series of the optimal policies under the discounted reward
criterion. The following Lemma shows that the optimal policy for the
single-armed bandit process with subsidy under the average reward
criterion is also a threshold policy.
\begin{lemma}\label{lemma:thresholdPolicyavg}
Let $\omega_{\beta}^*(m)$ denote the threshold of the optimal policy
for the single-armed bandit process with subsidy $m$ under the
discounted reward criterion. Then
$\lim_{\beta\rightarrow1}\omega^*_{\beta}(m)$ exists for any $m$.
Furthermore, the optimal policy for the single-armed bandit process with subsidy
$m$ under the average reward criterion is also a threshold policy
with threshold
$\omega^*(m)=\lim_{\beta\rightarrow1}\omega^*_{\beta}(m)$.
\end{lemma}
\begin{proof}
See Appendix E.
\end{proof}

\subsection{Indexability and Whittle's index
policy}\label{subsec:indexabilityavg}

Based on Lemma~\ref{lemma:thresholdPolicyavg}, the restless
multi-armed bandit process $(\Omega, \{{\bf P}_i\}_{i=1}^N,
\{B_i\}_{i=1}^N, 1)$ is indexable if the threshold $\omega^*(m)$ of
the optimal policy is monotonically increasing with subsidy $m$.
Next, we show that the monotonicity holds and the restless
multi-armed bandit process $(\Omega, \{{\bf P}_i\}_{i=1}^N,
\{B_i\}_{i=1}^N, 1)$ is indexable. Moreover, we obtain Whittle's
index in closed-form as shown below.

\vspace{1em}
\begin{theorem}\label{thm:whittleindex-timavg}
The restless multi-armed bandit process $(\Omega(1), \{{\bf
P}_i\}_{i=1}^N, \{B_i\}_{i=1}^N, 1)$ is indexable with Whittle's
index $W(\omega)$ given below.

\begin{itemize}
\item \emph{Case 1: Positively correlated channel $(p^{(i)}_{11}\ge p^{(i)}_{01})$.}
\end{itemize}
\begin{eqnarray}\label{eq:whittleindexavgP}
W(\omega)=\left\{\begin{array}{ll} \omega B_i, & \mbox{if}~\omega\le
p^{(i)}_{01}~\mbox{or}
~\omega\ge p^{(i)}_{11}\\[1em]
\frac{(\omega-\Tc^1(\omega))(L(p^{(i)}_{01},\omega)+1)+\Tc^{L(p^{(i)}_{01},\omega)}(p^{(i)}_{01})}{1-p^{(i)}_{11}+
(\omega-\Tc^1(\omega)L(p^{(i)}_{01},\omega)+\Tc^{L(p^{(i)}_{01},\omega)}(p^{(i)}_{01})}B_i,
 & \mbox{if}~p^{(i)}_{01}<\omega<\omega^{(i)}_o\\[1em]
\frac{\omega}{1-p^{(i)}_{11}+\omega}B_i, &
\mbox{if}~\omega^{(i)}_o\le
\omega<p^{(i)}_{11}\\[1em]
\end{array}\right..
\end{eqnarray}
\\[1em]

\begin{itemize}
\item \emph{Case 2: Negatively correlated channel $(p^{(i)}_{11}<p^{(i)}_{01})$.}
\end{itemize}
\begin{eqnarray}\label{eq:whittleindexavgN}
W(\omega)=\left\{\begin{array}{ll} \omega B_i, & \mbox{if}~\omega\le
p^{(i)}_{11}~\mbox{or}
~\omega\ge p^{(i)}_{01}\\[1em]
\frac{\omega+p^{(i)}_{01}-\Tc^1(\omega)}{1+p^{(i)}_{01}-\Tc^1(p^{(i)}_{11})+\Tc^1(\omega)-\omega}B_i
 & \mbox{if}~p^{(i)}_{11}<\omega<\omega^{(i)}_o\\[1em]
\frac{p^{(i)}_{01}}{1+p^{(i)}_{01}-\Tc^1(p^{(i)}_{11})}B_i, & \mbox{if}~\omega^{(i)}_o\le\omega<\Tc^1(p^{(i)}_{11})\\[1em]
\frac{p^{(i)}_{01}}{1+p^{(i)}_{01}-\omega}B_i, &
\mbox{if}~\Tc^1(p^{(i)}_{11})\le
\omega<p^{(i)}_{01}\\[1em]
\end{array}\right..
\end{eqnarray}
\end{theorem}

\begin{proof}
See Appendix F.
\end{proof}

\vspace{1em}

The monotonicity and piecewise concave/convex properties of
Whittle's index under the discounted reward criterion given in
Corollary~\ref{cor:Wproperty} are preserved under the average reward
criterion. The only difference is that Whittle's index under the
discounted reward criterion is always strictly increasing with the
belief state while Whittle's index $W(\omega)$ under the average
reward criterion is a constant function of $\omega$ when
$\omega_o\le\omega<\Tc^1(p_{11})$ for a negatively correlated channel (see~\eqref{eq:whittleindexavgN}).

\subsection{The Performance of Whittle's Index
Policy}\label{subsec:performanceavg}

Similar to the case under the discounted reward criterion, Whittle's
index policy is optimal under the average reward criterion when the
constraint on the number of activated arms $K(t)~(t\ge1)$ is relaxed
to the following.
\[\mathbb{E}_{\pi}[\lim_{T\rightarrow\infty}\frac{1}{T}\Sigma_{t=1}^{T}K(t)]=K.\]
Let $\bar{J}(\Omega(1))$ denote the maximum expected average reward
that can be obtained under this relaxed constraint when the initial
belief vector is $\Omega(1)$. Based on the Lagrangian multiplier
theorem, we have~\cite{whittle}
\begin{eqnarray}\label{eqn:relaxavg}
\bar{J}=\inf_{m}\{\Sigma_{i=1}^NJ^{(i)}_m-m(N-K)\},
\end{eqnarray}
where $J^{(i)}_m$ is the value function of the single-armed bandit
process with subsidy $m$ that
corresponds to the $i$-th channel.

Let $J(\Omega(1))$ denote the maximum expected average reward of the
RMBP under the strict
constraint that $K(t)=K$ for all $t$. Obviously,
\[J(\Omega(1))\le\bar{J}.\]
$\bar{J}$ thus provides a performance benchmark for Whittle's index
policy under the strict constraint. To evaluate $\bar{J}$, we
consider the single-armed bandit with subsidy $m$ under the average
reward criterion. The value function $J_m$ and the average
passive time $D_m=\frac{d(J_m)}{dm}$
can be obtained in closed-form as shown in
Lemma~\ref{lemma:passivetimeavg} below.

\begin{lemma}\label{lemma:passivetimeavg}
The value function $J_m$ and $D_m$ can be obtained in closed-form as
given below, where $\omega^*(m)$ is the threshold of the optimal
policy. Furthermore, $D_m$ is piecewise constant and increasing with $m$.
\begin{eqnarray}
J_m=\left\{\begin{array}{ll} \omega_o, & \mbox{if}~
\omega^*(m)<\min\{p_{01},p_{11}\}\\
\frac{(1-p_{11})L(p_{01},\omega^*(m))m+\Tc^{L(p_{01},\omega^*(m))}(p_{01})}{(1-p_{11})(L(p_{01},\omega^*(m))+1)
+\Tc^{L(p_{01},\omega^*(m))}(p_{01})}, & \mbox{if}~
p_{01}\le\omega^*(m)<\omega_o\\
\frac{p_{01}m+p_{01}}{1+2p_{01}-\Tc^1(p_{11})}, & \mbox{if}~
p_{11}\le\omega^*(m)<\Tc^1(p_{11})\\
m, & \mbox{other cases}
\end{array} \right.\label{eq:closeformJ}
\end{eqnarray}
and
\begin{eqnarray}
D_m=\left\{\begin{array}{ll} 0, & \mbox{if}~
\omega^*(m)<\min\{p_{01},p_{11}\}\\
\frac{(1-p_{11})L(p_{01},\omega^*(m))}{(1-p_{11})(L(p_{01},\omega^*(m))+1)+\Tc^{L(p_{01},\omega^*(m))}(p_{01})},
& \mbox{if}~
p_{01}\le\omega^*(m)<\omega_o\\
\frac{p_{01}}{1+2p_{01}-\Tc^1(p_{11})}, & \mbox{if}~
p_{11}\le\omega^*(m)<\Tc^1(p_{11})\\
1, & \mbox{other cases}
\end{array} \right..\label{eq:closeformD}
\end{eqnarray}
\end{lemma}

\begin{proof}
Under the value-boundedness condition as shown in
Sec.~\ref{subsec:vpavg}, we have, according to Dutta's theorem,
\[
J_m=\lim_{\beta_k\rightarrow1}(1-\beta_k)V_{\beta_k}(\omega,m),
\]
which leads to \eqref{eq:closeformJ} directly. The closed-form
expression for $D_m$ can be obtained from the derivative of $J_m$
with respect to $m$. The proof that $D_m$ is increasing with $m$ is
similar to that given in Lemma~\ref{lemma:passivetimeform}.
\end{proof}

Based on the closed-form $D_m$ given in
Lemma~\ref{lemma:passivetimeavg}, we can obtain the subsidy $m^*$
that achieves the infimum in~\eqref{eqn:relaxavg}. Specifically, the
subsidy $m^*$ that achieves the infimum in~\eqref{eqn:relaxavg} is
the supremum value of $m\in[0,1]$ satisfying $\Sigma_{i=1}^ND^{m,i}\le
N-K$. After obtaining $m^*$, it is easy to calculate the infimum
according to the closed-form $J_m$ given in
Lemma~\ref{lemma:passivetimeavg}. With minor changes, the algorithm
in Sec.~\ref{subsec:performancediscount} can be applied to evaluate
the upper bound $\bar{J}$. We notice that the initial belief will
not be considered in the algorithm, which leads to a shorter running
time.

Simulation results similarly to Fig.~\ref{fig:perform} have been observed,
demonstrating
the near-optimal performance of Whittle's index policy under the
average reward criterion .


\section{Whittle's Index Policy for Stochastically Identical
Channels}\label{sec:identical}

Based on the equivalency between Whittle's index policy and the myopic
policy for stochastically identical arms, we can analyze Whittle's index policy
by focusing on the myopic policy which has a much simpler index form.
In this section, we establish the semi-universal structure and study the
optimality of Whittle's index policy for stochastically identical arms.

\subsection{The Structure of Whittle's Index Policy}\label{subsec:structure}

The implementation of Whittle's index policy can be described with a
queue structure. Specifically, all $N$ channels are ordered in a
queue, and in each slot, those $K$ channels at the head of the queue
are sensed. Based on
the observations, channels are reordered at the end of each slot according
to the following simple rules.

When $p_{11}\ge p_{01}$, the channels observed in state $1$ will
stay at the head of the queue while the channels observed in state
$0$ will be moved to the end of the queue (see
Fig.~\ref{fig:Pstructure1}).

When $p_{11}<p_{01}$, the channels
observed in state $0$ will stay at the head of the queue while the
channels observed in state $1$ will be moved to the end of the queue.
The order of the unobserved channels are reversed (see Fig.~\ref{fig:Nstructure1}).

\begin{figure}[htb]
\begin{minipage}{4in}
\hspace{-2em}\leftline{
\begin{psfrags}
\psfrag{s}[l]{Sense} \psfrag{i1}[c]{{\footnotesize
$S_1(t)=1$}}\psfrag{i3}[c]{{\footnotesize$S_3(t)=1$}} \psfrag{b}[c]{
{\footnotesize$S_2(t)=0$}} \psfrag{t1}[c]{$\Kc(t)$}
\psfrag{t2}[c]{$\Kc(t+1)$}
\scalefig{0.8}\epsfbox{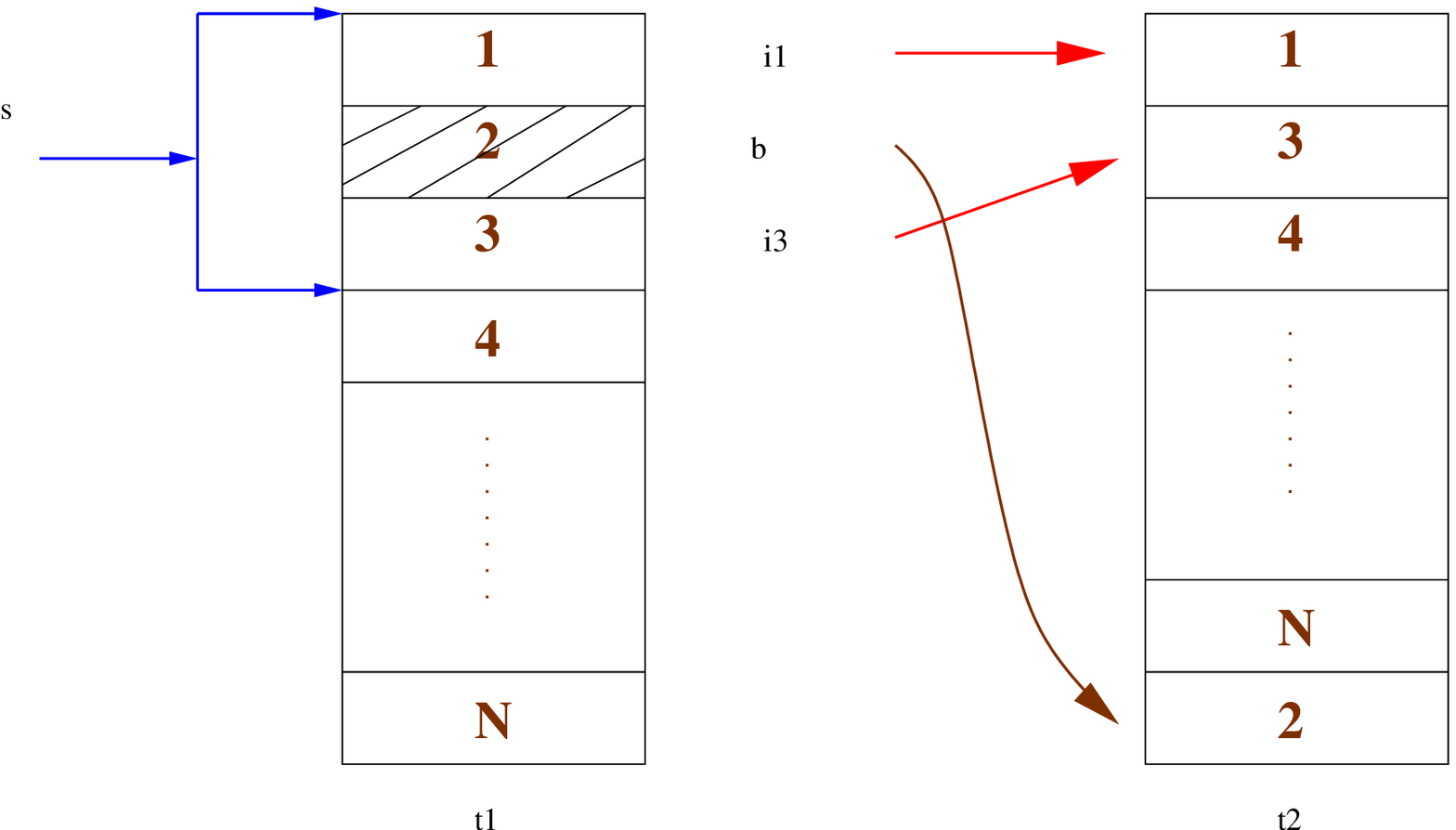}
\end{psfrags}}\caption{The structure of Whittle's index policy ($p_{11}\ge p_{01}$)}\label{fig:Pstructure1}
\end{minipage}\hspace{-4em}
\begin{minipage}{4in}
\hspace{-3em}\centerline{
\begin{psfrags}
\psfrag{s}[l]{Sense}
\psfrag{i1}[c]{{\footnotesize$S_1(t)=1$}}\psfrag{i3}[c]{{\footnotesize$S_3(t)=1$}}
\psfrag{b}[c]{ {\footnotesize$S_2(t)=0$}} \psfrag{t1}[c]{$\Kc(t)$}
\psfrag{t2}[c]{$\Kc(t+1)$}\psfrag{f}[l]{Flip}
\scalefig{0.9}\epsfbox{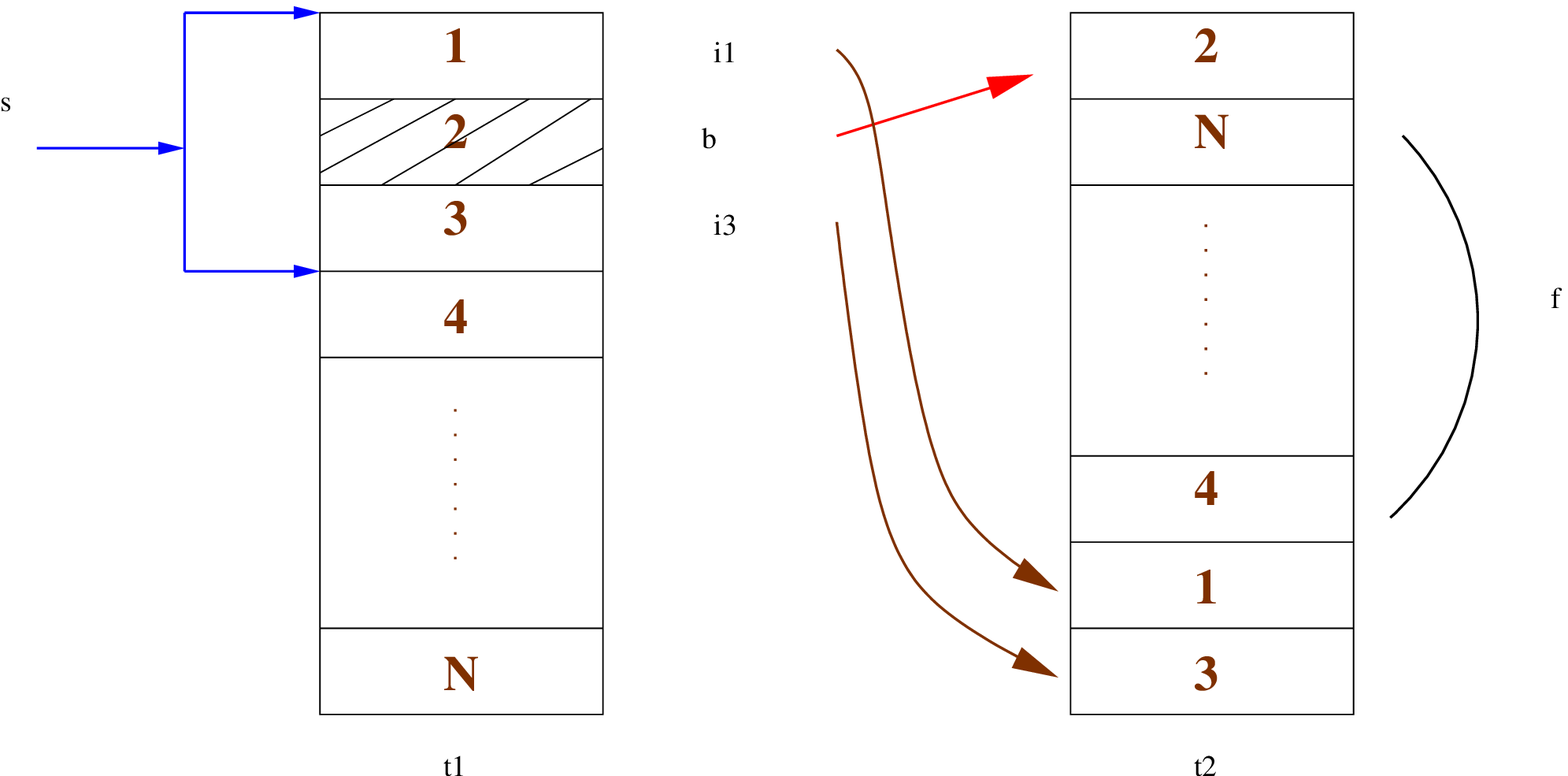}
\end{psfrags}}\caption{The structure of Whittle's index policy ($p_{11}<p_{01}$)}\label{fig:Nstructure1}
\end{minipage}
\end{figure}

The initial channel ordering $\Kc(1)$ is
determined by the initial belief vector as given below.
\begin{equation}
\omega_{n_1}(1)\ge\omega_{n_2}(1)\ge\cdots\ge\omega_{n_N}(1)
~\Longrightarrow ~\Kc(1)=(n_1,n_2,\cdots,n_N). \label{eq:Kc}
\end{equation}

See Appendix G
for the proof of the structure of Whittle's index policy.


The advantage of this structure of Whittle's index policy is twofold.
First, it demonstrates the simplicity of Whittle's index policy:
channel selection is reduced to maintaining a simple queue structure that
requires no computation and little memory. Second, it
shows that Whittle's index policy has a semi-universal structure;
it can be implemented without knowing the channel transition
probabilities except the order of $p_{11}$ and $p_{01}$. As a result, Whittle's index
policy is robust against model mismatch and automatically tracks variations in the channel model provided
that the order of $p_{11}$ and $p_{01}$ remains unchanged. As show
in Fig.~\ref{fig:myopictraking}, the transition probabilities change
abruptly in the fifth slot, which corresponds to an increase in the occurrence of good
channel state in the system. From this figure, we can
observe, from the change in the throughput increasing rate, that
Whittle's index policy effectively tracks the model variations.

\begin{figure}[h]
\centerline{
\begin{psfrags}
\scalefig{0.6}\epsfbox{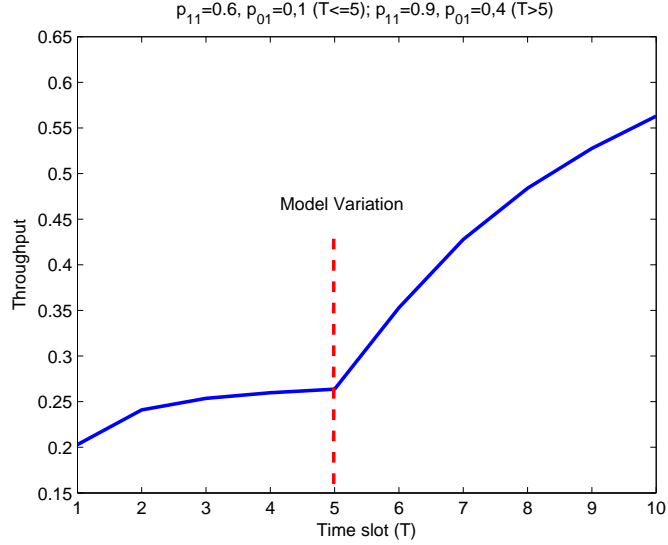}
\end{psfrags}}
\caption{Tracking the change in channel transition probabilities
occurred at $t=6$.}\label{fig:myopictraking}
\end{figure}

\subsection{Optimality and Approximation Factor of Whittle's Index Policy}\label{subsec:worstcase}
Based on the simple structure of Whittle's index policy for
stochastically identical channels, we can obtain a lower bound of
its performance. Combining this lower bound and the upper bound
shown in Sec.~\ref{subsec:performanceavg}, we further obtain the
approximation factor of the performance by Whittle's index policy,
which are independent of channel parameters. Recall that
$J$ denote the average reward achieved by the optimal
policy. Let $J_w$ denote the
average reward achieved by Whittle's
index policy,

\begin{theorem}\label{thm:performanceidentical} {\it Lower and Upper Bounds of The Performance of Whittle's Index Policy}
\begin{eqnarray}\label{eqn:optimalbound}
\frac{K\Tc^{\lfloor\frac{N}{K}\rfloor-1}(p_{01})}{1-p_{11}+\Tc^{\lfloor\frac{N}{K}\rfloor-1}(p_{01})}\le
J_w\le J\le\min\{\frac{K\omega_o}{1-p_{11}+\omega_o},~\omega_oN\}
~~~~~~&\mbox{if}~~p_{11}\ge p_{01}\\\label{eqn:optimalbound1}
\frac{Kp_{01}}{1-\Tc^{2\lfloor\frac{N}{K}\rfloor-2}(p_{11})+p_{01}}\le
J_w\le
J\le\min\{\frac{Kp_{01}}{1-\Tc^1(p_{11})+p_{01}},~\omega_oN\}&\mbox{if}~~p_{11}<p_{01}
\end{eqnarray}
\end{theorem}

\begin{proof}
The upper bound of $J$ is obtained from the upper bound of the
optimal performance for generally non-identical channels as given
in~\eqref{eqn:relaxavg}. The lower bound of $J_w$ is obtained from the
structure of Whittle's index policy. See Appendix H for the complete
proof.
\end{proof}

\vspace{1em}

\begin{corollary}\label{cor:appfac}
Let $\eta=\frac{J_w}{J}$ be the approximation factor defined as the
ratio of the performance by Whittle's index policy to the optimal
performance. We have\\

\begin{minipage}{4in}
\hspace{2em}Positively correlated channels\\[-0.5em]
\[\hspace{-6em}\left\{\begin{array}{ll} \eta=1,~~\mbox{for}~K=1,N-1,N\\
\eta\ge\frac{K}{N},~~\mbox{o.w.}\end{array}\right.,\]
\end{minipage}
\begin{minipage}{4in}
\hspace{-2em}Negatively correlated channels\\[-0.5em]
\[\hspace{-14em}\left\{\begin{array}{ll} \eta=1,~~\mbox{for}~K=N-1,N\\
\eta\ge\max\{\frac{1}{2},\frac{K}{N}\},~~\mbox{o.w.}\end{array}\right..\]
\end{minipage}\\[0.5em]
\end{corollary}

\begin{proof}
See Appendix I.
\end{proof}

Fig.~\ref{fig:consfac} illustrates the approximation factors of
Whittle's index policy for both positively correlated and negatively
correlated channels. We notice that the approximation factor
approaches to $1$ as $K$ increases. For negatively correlated
channels, Whittle's index policy achieves at least half the optimal
performance. For positively correlated channels, the approximation
factor can be further improved under certain conditions on the transition
probabilities. Specifically, we have
$\eta\ge 1-p_{11}+\omega_o$.

From Corollary~\ref{cor:appfac}, Whittle's index policy is optimal
when $K=1~\mbox{(for positively correlated
channels)}$ and $K=N-1$. The optimality for $K=N$ is trivial.
We point out that for a general $K$, numerical examples have shown
that actions given by Whittle's index policy match with the optimal
actions for randomly generated sample paths, suggesting the optimality of
Whittle's
index policy.

\begin{figure}
\hspace{-3em}\begin{minipage}{4in} \centerline{
\begin{psfrags}
\scalefig{0.9}\epsfbox{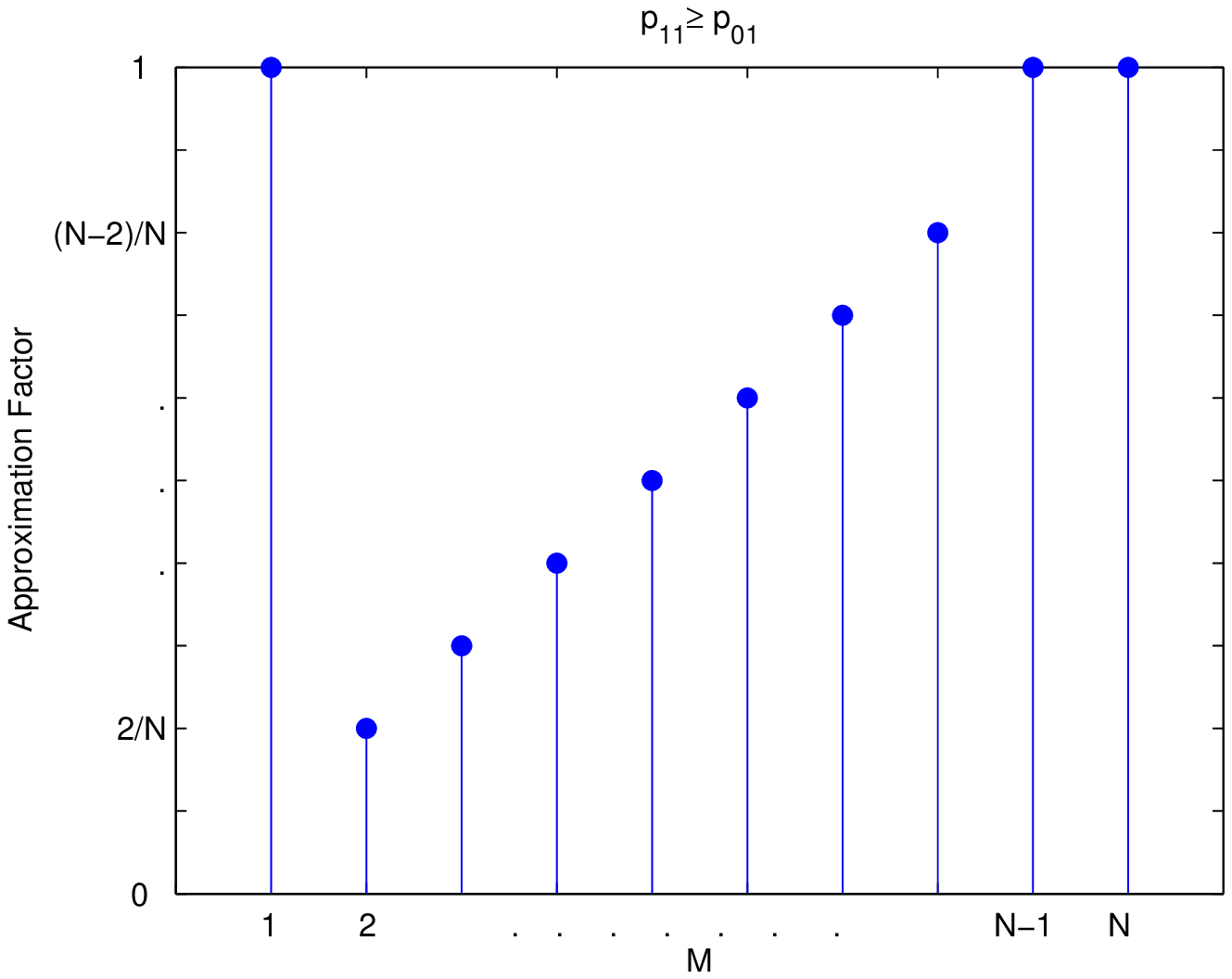}
\end{psfrags}}
\end{minipage}
\hspace{-1em}
\begin{minipage}{4in}
\hspace{-1em}\centerline{
\begin{psfrags}
\scalefig{0.9}\epsfbox{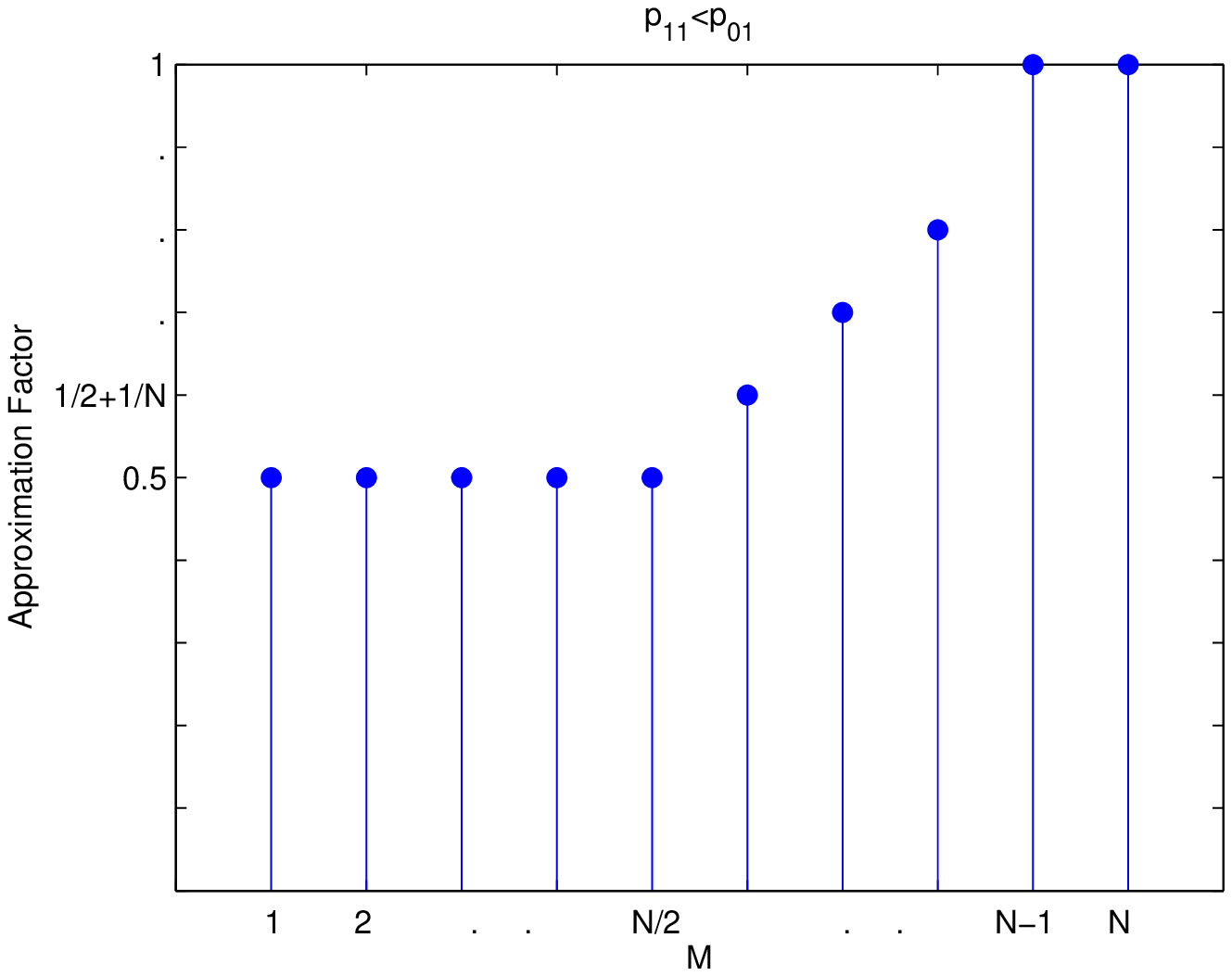}
\end{psfrags}}
\end{minipage}\caption{The approximation factor of Whittle's index
policy.}\label{fig:consfac}
\end{figure}

\section{Conclusion} \label{sec:conclusion}

In this paper, we have formulated the multi-channel opportunistic
access problem as a restless multi-armed bandit process. We
established the indexability and obtained Whittle's index in
closed-form under both discounted and average reward criteria. We
developed efficient algorithms for computing an upper bound of the
optimal policy, which is the optimal performance under the relaxed
constraint on the average number of channels that can be sensed
simultaneously. When channels are stochastically identical, we have
shown that Whittle's index policy coincides with the myopic policy.
Based on this equivalency, we have established the semi-universal
structure and the optimality of Whittle index policy under certain
conditions.


\section*{Appendix A: Proof of Lemma \ref{lemma:valuefunction}}\label{sec:A}

From~\eqref{eq:vcloseform}, we have
\begin{eqnarray}\label{eq:vp01closeform}
V_{\beta,m}(p_{01})=
\frac{1-\beta^{L(p_{01},\omega_{\beta}^*(m))}}{1-\beta}m+\beta^{L(p_{01},\omega_{\beta}^*(m))}V_{\beta,m}
(\Tc^{L(p_{01},\omega_{\beta}^*(m))}(p_{01});u=1),\\\label{eq:vp11closeform}
V_{\beta,m}(p_{11})=
\frac{1-\beta^{L(p_{11},\omega_{\beta}^*(m))}}{1-\beta}m+\beta^{L(p_{11},\omega_{\beta}^*(m))}V_{\beta,m}
(\Tc^{L(p_{11},\omega_{\beta}^*(m))}(p_{01});u=1).
\end{eqnarray}
As shown in~\eqref{eqn:value_a},
$V_{\beta,m}(\Tc^{L(\omega,\omega_{\beta}^*(m))}(\omega);u=1)$ is a
function of $V_{\beta,m}(p_{01})$ and $V_{\beta,m}(p_{11})$ for any
$\omega\in[0,1]$. We thus have two
equations~\eqref{eq:vp01closeform} and~\eqref{eq:vp11closeform} for
two unknowns $V_{\beta,m}(p_{01})$ and $V_{\beta,m}(p_{11})$
provided that we can obtain the two crossing times
$L(p_{01},\omega^*_{\beta}(m))$ and $L(p_{11},\omega^*_{\beta}(m))$.

From~\eqref{eqn:Lpositive} and~\eqref{eqn:Lnegative}, we can obtain
these crossing times by considering different regions that the
threshold $\omega^*_{\beta}(m)$ may lie in (see
Fig.~\ref{fig:activetime} and Fig.~\ref{fig:activetime1}). We can
thus solve for $V_{\beta,m}(p_{01})$ and $V_{\beta,m}(p_{11})$
from~\eqref{eq:vp01closeform} and~\eqref{eq:vp11closeform} by
considering each region within which both crossing times
$L(p_{01},\omega^*_{\beta}(m))$ and $L(p_{11},\omega^*_{\beta}(m))$
are constant.

\begin{figure}[h]
\centerline{
\begin{psfrags}
\psfrag{0}[c]{$0$}\psfrag{1}[c]{$1$}
\psfrag{p0}[c]{$p_{01}$}\psfrag{p1}[c]{$p_{11}$}
\psfrag{w0}[c]{$\omega_o$} \psfrag{a}[c]{$L=0$}\psfrag{p}[c]{
$L=\infty$} \psfrag{al}[c]{~~~\small
$L=\lfloor\log_{p_{11}-p_{01}}^{\frac{p_{01}-\omega_{\beta}^*(m)(1-p_{11}+p_{01})}{p_{01}(p_{11}-p_{01})}}\rfloor+1$
}\psfrag{sp0}[c]{$L=L(p_{01},\omega^*_{\beta}(m))$~~~~~~~~}
\psfrag{sp1}[c]{ $L=L(p_{11},\omega^*_{\beta}(m))$~~~~~~~~~}
\psfrag{w}[l]{$\omega^*_{\beta}(m)$}
\scalefig{0.95}\epsfbox{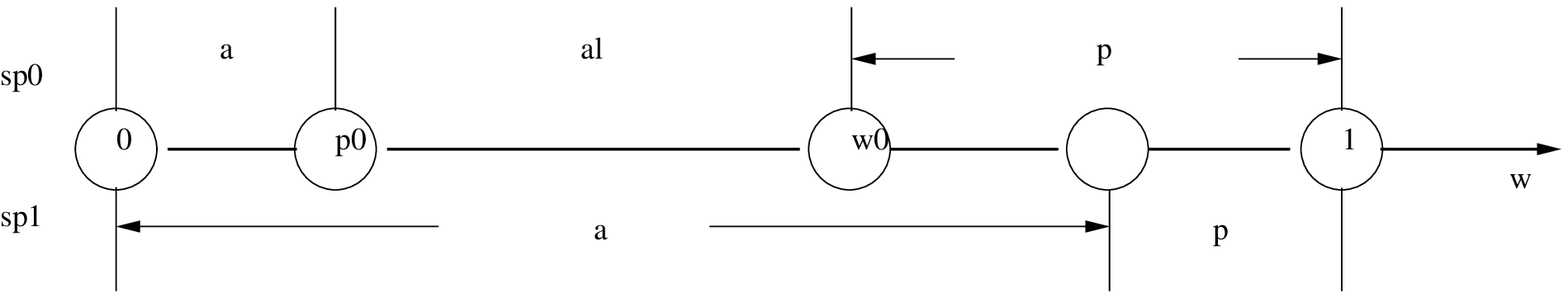}
\end{psfrags}}\caption{The threshold crossing time for different regions of
$\omega^*_{\beta}(m)$ when $p_{11}\ge p_{01}$ (the top partition is
for $L(p_{01},\omega^*_{\beta}(m))$, the bottom for
$L(p_{11},\omega^*_{\beta}(m))$).}\label{fig:activetime}
\end{figure}
\begin{figure}[h]
\centerline{
\begin{psfrags}
\psfrag{0}[c]{$0$}\psfrag{1}[c]{$1$}
\psfrag{p0}[c]{$p_{01}$}\psfrag{p1}[c]{$p_{11}$}
\psfrag{tp1}[c]{\footnotesize $\Tc(p_{11})$}
\psfrag{a}[c]{~~$L=0$}\psfrag{p}[c]{$L=\infty$}
\psfrag{al}[c]{~~~$L=1$
}\psfrag{sp0}[c]{$L=L(p_{11},\omega^*_{\beta}(m))$~~~~~~~~}
\psfrag{sp1}[c]{$L=L(p_{01},\omega^*_{\beta}(m))$~~~~~~~~}
\psfrag{w}[l]{$\omega^*_{\beta}(m)$}
\scalefig{0.95}\epsfbox{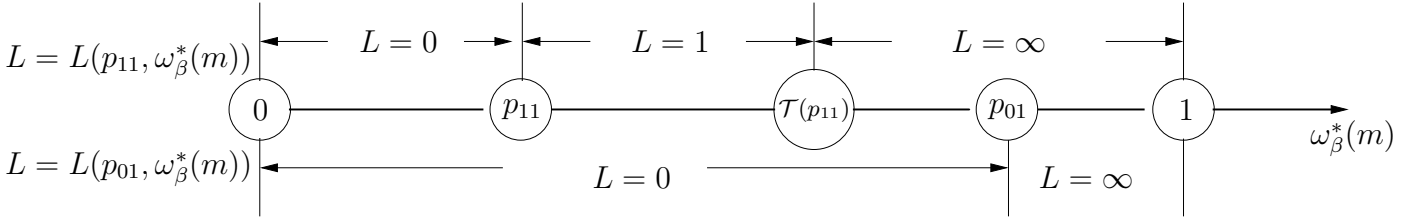}
\end{psfrags}}\caption{The threshold crossing time for different regions of
$\omega^*_{\beta}(m)$ when $p_{11}<p_{01}$ (the top partition is for
$L(p_{11},\omega^*_{\beta}(m))$, the bottom for
$L(p_{01},\omega^*_{\beta}(m))$).}\label{fig:activetime1}
\end{figure}

\section*{Appendix B: Proof of Theorem~\ref{thm:indexability}}\label{sec:B}
It suffices to prove that an arm with an arbitrary transition matrix
$\bf{P}$ is indexable. Based on the threshold structure of the
optimal policy for the single-armed bandit with subsidy $m$ given in
Lemma~\ref{lemma:thresholdPolicy}, indexability is reduced to the
monotonicity of the threshold $\omega^*_{\beta}(m)$, \ie
$\omega^*_{\beta}(m)$ is monotonically increasing with the subsidy
$m$ for $m\in[0,1)$. To prove the monotonicity of
$\omega^*_{\beta}(m)$, we first give the following lemma.

\begin{lemma}\label{lemma:condindex}
Suppose that for any $m\in[0,1)$ we have
\begin{eqnarray}\label{pf:index1}
\frac{dV_{\beta,m}(\omega;u=1)}{dm}|_{\omega=\omega_{\beta}^*(m)}<
\frac{dV_{\beta,m}(\omega;u=0)}{dm}|_{\omega=\omega_{\beta}^*(m)}.
\end{eqnarray}
Then $\omega^*_{\beta}(m)$ is monotonically increasing with $m$.
\end{lemma}

We prove Lemma~\ref{lemma:condindex} by contradiction. Assume that there
exists an $m_0\in[0,1)$ such that $\omega^*_{\beta}(m)$ is
decreasing at $m_0$. Then, there exists an $\epsilon>0$ such that
for any $\Delta m\in[0,\epsilon]$, we have
\begin{eqnarray}\label{eq:Vincrease}
V_{\beta,m_0+\Delta
m}(\omega^*_{\beta}(m_0);u=1)\ge~V_{\beta,m_0+\Delta
m}(\omega^*_{\beta}(m_0);u=0).
\end{eqnarray}
Since $\omega^*_{\beta}(m_0)$ is the threshold of the optimal policy
under subsidy $m_0$, we have
\begin{eqnarray}\label{eq:Vequal}
V_{\beta,m_0}(\omega^*_{\beta}(m_0);u=1)=V_{\beta,m_0}(\omega^*_{\beta}(m_0);u=0).
\end{eqnarray}
From~\eqref{eq:Vincrease} and~\eqref{eq:Vequal}, we have
\[\frac{dV_{\beta,m}(\omega;u=1)}{dm}|_{\omega=\omega_{\beta}^*(m_0)}\ge
\frac{dV_{\beta,m}(\omega;u=0)}{dm}|_{\omega=\omega_{\beta}^*(m_0)},\]
which contradicts with~\eqref{pf:index1}.
Lemma~\ref{lemma:condindex} thus holds.

According to Lemma~\ref{lemma:condindex}, it is sufficient to
prove~\eqref{pf:index1}. Recall that
$D_{\beta,m}(\omega)=\frac{d(V_{\beta,m}(\omega))}{dm}$.
From~\eqref{eqn:value_a} and~\eqref{eqn:value_b}, we can
write~\eqref{pf:index1} as
\begin{eqnarray}
\beta
(\omega_{\beta}^*(m)D_{\beta,m}(p_{11})+(1-\omega_{\beta}^*(m))D_{\beta,m}(p_{01}))<
1+\beta
D_{\beta,m}(\Tc^1(\omega_{\beta}^*(m))).\label{eqn:passiveinequal}
\end{eqnarray}

To prove \eqref{eqn:passiveinequal}, we consider the following three
regions of $\omega_{\beta}^*(m)$.

\textbf{Region 1:} $0\le\omega_{\beta}^*(m)< \min\{p_{01},p_{11}\}.$

Based on the lower bound of the updated belief given in
Lemma~\ref{lemma:ksteptransition}, the arm will be activated in
every slot when the initial belief $\omega>\omega_{\beta}^*(m)$.
Thus,
$D_{\beta,m}(p_{11})=D_{\beta,m}(p_{01})=D_{\beta,m}(\Tc^1(\omega_{\beta}^*(m)))=0$;
\eqref{eqn:passiveinequal} holds trivially.

\textbf{Region 2:} $\omega_o\le\omega_{\beta}^*(m)\le 1$.

In this region, the arm is made passive in every slot when the
initial belief state is $\Tc^1(\omega_{\beta}^*(m))$. This is
because $\Tc^k(\omega_{\beta}^*(m))\le\omega_{\beta}^*(m)$ for any
$k\ge1$ (see Lemma~\ref{lemma:ksteptransition},
Fig.~\ref{fig:ksteppositive} and Fig.~\ref{fig:kstepnegative}).
Therefore,
$D_{\beta,m}(\Tc^1(\omega_{\beta}^*(m)))=\frac{1}{1-\beta}$. Since
both $D_{\beta,m}(p_{11})$ and $D_{\beta,m}(p_{01})$ are upper
bounded by $\frac{1}{1-\beta}$, it is easy to see that
\eqref{eqn:passiveinequal} holds.

\textbf{Region 3:}
$\min\{p_{01},p_{11}\}\le\omega_{\beta}^*(m)<\omega_o$.

In this region, $\Tc^1(\omega_{\beta}^*(m))>\omega_{\beta}^*(m)$
(see Fig.~\ref{fig:ksteppositive} and Fig.~\ref{fig:kstepnegative}).
Thus, $\Tc^1(\omega_{\beta}^*(m))$ is in the active set, which gives
us
\begin{eqnarray}\label{eqn:d1}
D_{\beta,m}(\Tc^1(\omega_{\beta}^*(m)))&=&\beta(\Tc^1(\omega_{\beta}^*(m))
D_{\beta,m}(p_{11})+(1-\Tc^1(\omega_{\beta}^*(m)))D_{\beta,m}(p_{01}))
\end{eqnarray}
To prove~\eqref{eqn:passiveinequal}, we consider the positively
correlated and negatively correlated cases separately.
\begin{itemize}
\item \emph{Case 1: Negatively correlated channel ($p_{11}<p_{01}$).}
\end{itemize}
Since $p_{01}>\omega_o>\omega_{\beta}^*(m)$, $p_{01}$ is in the
active set. We thus have {\begin{eqnarray}\label{eqn:dd2}
D_{\beta,m}(p_{01})&=&\beta(p_{01}
D_{\beta,m}(p_{11})+(1-p_{01})D_{\beta,m}(p_{01})).
\end{eqnarray}}
Substituting~\eqref{eqn:d1} and~\eqref{eqn:dd2}
into~\eqref{eqn:passiveinequal}, we reduce
~\eqref{eqn:passiveinequal} to the following inequality.
\begin{eqnarray}\label{eqn:p01ineq}
\frac{\beta}{1-\beta(1-p_{01})}D_{\beta,m}(p_{11})(1-\beta)(\beta
p_{01}+\omega_{\beta}^*(m)-\beta\Tc^1(\omega_{\beta}^*(m)))<1.
\end{eqnarray}
Notice that the left-hand side of~\eqref{eqn:p01ineq} is increasing
with $\omega_{\beta}^*(m)$ and $D_{\beta,m}(p_{11})$. It thus
suffices to show the inequality by replacing $\omega_{\beta}^*(m)$
with its upper bound $\omega_o$ and $D_{\beta,m}(p_{11})$ with its
upper bound $\frac{1}{1-\beta}$. After some simplifications, it is
sufficient to prove
\begin{eqnarray}
f(\beta)\defeq
p_{01}(p_{01}-p_{11})\beta^2+\beta(p_{01}+1-p_{11}-p_{01}^2+p_{01}p_{11})-1-p_{01}+p_{11}<0.
\end{eqnarray}
It is easy to see that $f(\beta)$ is convex in $\beta$,
$f(0)=-1-p_{01}+p_{11}<0$, and $f(1)=0$. We thus conclude that
$f(\beta)<0$ for any $0\le\beta<1$.

\begin{itemize}
\item \emph{Case 2: Positively correlated channel ($p_{11}>p_{01}$).}
\end{itemize}

Since $p_{11}\ge\omega_o>\omega_{\beta}^*(m)$, $p_{11}$ is in the
active set. We thus have
\begin{eqnarray}\label{eqn:d2}
D_{\beta,m}(p_{11})&=&\beta(p_{11}
D_{\beta,m}(p_{11})+(1-p_{11})D_{\beta,m}(p_{01})).
\end{eqnarray}

Substituting~\eqref{eqn:d1} and \eqref{eqn:d2}
into~\eqref{eqn:passiveinequal}, we
reduce~\eqref{eqn:passiveinequal} to the following inequality.
\begin{eqnarray}\label{eqn:key1}
\beta D_{\beta,m}(p_{01})
(1-\beta)(1-\frac{\omega_{\beta}^*(m)-\beta\Tc^1(\omega_{\beta}^*(m))}{1-\beta
p_{11}})<1.
\end{eqnarray}

Substituting the closed-form of $D_{\beta,m}(p_{01})$ given
in~\eqref{eq:PcloseformDp01} into~\eqref{eqn:key1}, we end up with
an inequality in terms of $L(p_{01},\omega^*_{\beta}(m))$ and
$\omega^*_{\beta}(m)$. Notice that the left-hand side
of~\eqref{eqn:key1} is decreasing with $\omega^*_{\beta}(m)$. It
thus suffices to show the inequality by replacing
$\omega^*_{\beta}(m)$ with its lower bound
$\Tc^{L(p_{01},\omega^*_{\beta}(m))-1}(p_{01})$ (by the definition
of $L(p_{01},\omega^*_{\beta}(m))$). Let $x=p_{11}-p_{01}$. After
some simplifications, it is sufficient to show that for any
$0\le\beta<1,~0\le p_{01}\le 1,~0\le x\le 1-p_{01},~
L\in\{0,1,2,..\}$, {\small\begin{eqnarray}
f(\beta)\defeq\beta^{L+2}p_{01} x^{L+1}(1-x)+\beta^2(p_{01}
x^{L+2}+x-x^2-p_{01} x)+\beta(x^2+p_{01} x-p_{01} x^{L+1}-1)+1-x>0.
\end{eqnarray}}
Since $f(0)=1-x>0$ and $f(1)=0$, it is sufficient to prove that
$f(\beta)$ is strictly decreasing with $\beta$ for $0\le\beta\le1$,
which follows by showing $\frac{d(f(\beta))}{d(\beta)}<0$ for
$0\le\beta<1$. {\small\begin{eqnarray}\label{eqn:diff}
\frac{d(f(\beta))}{d(\beta)}&=&(L+2)\beta^{L+1}p_{01}
x^{L+1}(1-x)+2\beta(p_{01} x^{L+2}+x-x^2-p_{01} x)+(x^2+p_{01}
x-p_{01} x^{L+1}-1).~~~~~~~
\end{eqnarray}}
To show $\frac{d(f(\beta))}{d(\beta)}<0$ for $0\le\beta<1$, we will
establish the following two facts:
\begin{itemize}
\item[(i)] $\frac{d(f(\beta))}{d(\beta)}|_{\beta=1}\le
0$.\\[-1em]
\item[(ii)] $\frac{d(f(\beta))}{d(\beta)}$ is strictly increasing
with $\beta$.
\end{itemize}
To prove (i), we set $\beta=1$ in~\eqref{eqn:diff}. After some
simplifications, we need to prove {\begin{eqnarray}\label{eqn:inq}
h(p_{01})\defeq-p_{01} Lx^{L+2}+p_{01}(L+1)x^{L+1}-x^2-p_{01}
x+2x-1\le 0.
\end{eqnarray}}
Since  $h(0)=-(x-1)^2\le 0$, it is sufficient to prove that
$h(p_{01})$ is monotonically decreasing with $p_{01}$, \ie we need
to prove
\begin{eqnarray}\label{eqn:diff1}
\frac{d(h(p_{01}))}{d(p_{01})}=-Lx^{L+2}+(L+1)x^{L+1}-x\le0.
\end{eqnarray}
Since $Lx^{L+1}\le\Sigma_{k=1}^{L}x^k=\frac{x-x^{L+1}}{1-x}$, it is
easy to see that~\eqref{eqn:diff1} holds. We thus proved (i).

To prove (ii), it suffices to show that the coefficient of $\beta$
in~\eqref{eqn:diff} is nonnegative, \ie we need to prove
\begin{eqnarray}\label{eqn:diff2}
x^{L+2}+x-x^2-p_{01} x\ge0.
\end{eqnarray}
Since $0\le x\le1-p_{01}$, we have $p_{01}x(x^{L+1}-1)\ge
-p_{01}x\ge(x-1)x$. It is easy to see that~\eqref{eqn:diff2} holds.
We thus proved (ii).

From (i) and (ii), it is easy to see that
$\frac{d(f(\beta))}{d(\beta)}<0$ for any $0\le\beta<1$. We thus
proved the indexability.

\section*{Appendix C: Proof of Theorem~\ref{thm:algorithm}}\label{sec:C}

We notice that Step 1 runs in $O(N)$ time. In Step 2, the number
of regions that needs to be calculated for each channel is at most
$O(\log\frac{\delta}{N})=O(\log N)$. It runs in constant time to
find $l_i$ and $d_i$ for channel $i$. So Step 2 runs in at most
$O(N\log N)$ time. In Step 3, the ordering of all those
probabilities needs at most $O(N\log N)(\log(O(N\log N)))=O(N(\log
N)^2)$ time. Step 4 runs in $O(N)$ time for each region that does
not belong to $V$. So Step 4 runs in at most $O(N^2\log N)$ time.
Finally, Step $5$ runs in $O(N)$ time. Overall, the algorithm runs
in at most $O(N^2\log N)$ time.

\section*{Appendix D: Proof of Lemma~\ref{lemma:valueboundedness}}\label{sec:D}

From the closed-form $V_{\beta,m}(p_{01})$ (see Lemma
\ref{lemma:valuefunction}), we have, for any $\beta~(0\le\beta<1)$,
\begin{eqnarray}
|V_{\beta,m}(p_{01})-V_{\beta,m}(p_{11})|\le c.
\end{eqnarray}

From Fig.~\ref{fig:valuefunc2}, Fig.~\ref{fig:valuefunc1}, and
Fig.~\ref{fig:valuefunc}, we have, for any $\omega\in[0,1]$,
\begin{eqnarray}
\min\{V_{\beta,m}(0;u=1),V_{\beta,m}(1;u=1)\}\le
V_{\beta,m}(\omega)\le\max\{V_{\beta,m}(0;u=0),V_{\beta,m}(1;u=1)\}.
\end{eqnarray}

Consequently, we have, for any $\omega,\omega'\in [0,1]$,
{\small\begin{eqnarray}\nn
&&|V_{\beta,m}(\omega)-V_{\beta,m}(\omega')|\\\nn
&&\le\max(|V_{\beta,m}(0;u=1)-V_{\beta,m}(1;u=1)|,|V_{\beta,m}(0;u=0)-V_{\beta,m}(0;u=1)|,|V_{\beta,m}(0;u=0)
-V_{\beta,m}(1;u=1)|)\\\nn &&=\max(|\beta
(V_{\beta,m}(p_{01})-V_{\beta,m}(p_{11}))-1|, |\beta
(V_{\beta,m}(p_{01})-V_{\beta,m}(p_{11}))|, 1).
\end{eqnarray}}
Since $|V_{\beta,m}(p_{01})-V_{\beta,m}(p_{11})|\le c$ for any
$\beta~(0\le\beta<1)$, then
$V_{\beta,m}(\omega)-V_{\beta,m}(\omega')|\le c+1$ for any
$\beta~(0\le\beta<1)$ and $\omega,\omega'\in [0,1]$. Thus the
value-boundedness condition is satisfied.

\section*{Appendix E: Proof of Lemma~\ref{lemma:thresholdPolicyavg}}\label{sec:E}
The convergence of $\omega^*_{\beta}(m)$ is trivial for $m<0$ and
$m\ge 1$.

For $0\le m<1$, let
$W(\omega)=\lim_{\beta\rightarrow1}W_{\beta}(\omega)$. This limit
exists and is given in Theorem~\ref{thm:whittleindex-timavg} (it is
tedious and lengthy to get the limit and we skip the detailed
calculation). Define $\omega^*(m)$ as the inverse function of
$W(\omega)$. We notice that $W(\omega)$ is a constant function (thus
not invertible) when $\omega_o\le\omega\le\Tc^1(p_{11})$
(see~\eqref{eq:whittleindexavgN}). In this case, we set
$\omega^*(m)=\Tc^1(p_{11})$. Formally, we have
\begin{eqnarray}
\omega^*(m)=\left\{\begin{array}{c} c~(c<0)~~~~~~~~~~~~\mbox{if}~~m<0\\
\max\{\omega:W(\omega)=m\}~~~\mbox{if}~~0\le m<1\\
b~(b>1)~~~~~~~~~~~~~\mbox{if}~~m\ge
1\end{array}\right..\label{eqn:thresholdavg}
\end{eqnarray}
Next, we prove that
$\lim_{\beta\rightarrow1}\omega^*_{\beta}(m)=\omega^*(m)$ as
$\beta\rightarrow1$ by contradiction. Since
$W(\omega)=\lim_{\beta\rightarrow1}W_{\beta}(\omega)$ and
$W_{\beta}(\omega)$ is increasing with $\omega$, $W(\omega)$ is also
increasing with $\omega$. Assume first that $W_{\beta}(\omega)$ is
strictly increasing at point $\omega^*_{\beta}(m)$. We prove
$\lim_{\beta\rightarrow1}\omega^*_{\beta}(m)=\omega^*(m)$ by
contradiction as follows.


Assume $\omega_{\beta}^*(m)\nrightarrow\omega^*(m)$, \ie there
exists an $\epsilon>0$, a $\beta'~(0\le\beta'<1)$, and a series
$\{\beta_k\}~(\beta_k\rightarrow 1)$ such that
$|\omega_{\beta_k}^*(m)-\omega^*(m)|>\epsilon$ for any
$\beta_k>\beta'$. If $\omega^*(m)-\epsilon>\omega_{\beta_k}^*(m)$
for any $\beta_k>\beta'$, then $W_{\beta_k}(\omega^*(m)-\epsilon)\ge
W_{\beta_k}(\omega_{\beta_k}^*(m))$ for any $\beta_k>\beta'$ by the
monotonicity of $W_{\beta_k}(\omega)$. Since $W(\omega)$ is strictly
increasing at point $\omega^*(m)$, there exists a $\delta>0$ such
that $W(\omega^*(m))-W(\omega^*(m)-\epsilon)>\delta$. Then we have,
for any $\beta_k>\beta'$,
\[W_{\beta_k}(\omega^*(m)-\epsilon)\ge
W_{\beta_k}(\omega_{\beta_k}^*(m))=m=W(\omega^*(m))>W(\omega^*(m)-\epsilon)+\delta,\]
which contradicts with the fact that
$W_{\beta_k}(\omega_{\beta_k}^*(m)-\epsilon)\rightarrow
W(\omega^*(m)-\epsilon)$ as $\beta_k\rightarrow 1$. The proof for
the case when $\omega^*(m)+\epsilon<\omega_{\beta_k}^*(m)$ for any
$\beta_k>\beta'$ is similar to the above.

Consider next that $W(\omega)$ is not strictly increasing at
point $\omega^*(m)$. This case only occurs when $p_{11}<p_{01}$ and
$\omega^*(m)=\Tc^1(p_{11})$. We notice that
$W_{\beta}(\Tc^1(p_{11}))$ increasingly converges to
$W(\Tc^1(p_{11}))$ as $\beta\rightarrow 1$. Thus
$\omega_{\beta}^*(m)\ge\Tc^1(p_{11})=\omega^*(m)$ by the
monotonicity of $W_{\beta}(\omega)$. Assume
$\omega_{\beta}^*(m)\nrightarrow\omega^*(m)$, \ie, there exist an
$\epsilon>0$, a $\beta'~(0\le\beta'<1)$ and a series
$\{\beta_k\}~(\beta_k\rightarrow 1)$ such that
$\omega_{\beta_k}^*(m)-\omega^*(m)>\epsilon$ for any
$\beta_k>\beta'$. We have $W_{\beta_k}(\omega^*(m)+\epsilon)<
W_{\beta_k}(\omega_{\beta_k}^*(m))$ for any $\beta_k>\beta'$ by the
monotonicity of $W_{\beta_k}(\omega)$. Since $W(\omega)$ is strictly
increasing in $[\omega^*(m),\omega^*(m)+\epsilon]$, there exists a
$\delta'>0$ such that
$W(\omega^*(m)+\epsilon)-W(\omega^*(m))>\delta'$. Then we have, for
any $\beta_k>\beta'$, \[W_{\beta_k}(\omega^*(m)+\epsilon)\le
W_{\beta_k}(\omega_{\beta_k}^*(m))=m=W(\omega^*(m))<W(\omega^*(m)+\epsilon)-\delta',\]
which contradicts with the fact that
$W_{\beta_k}(\omega_{\beta_k}^*(m)+\epsilon)\rightarrow
W(\omega^*(m)+\epsilon)$ as $\beta_k\rightarrow 1$.

Next, we show that the optimal policy $\pi_{\beta_k}^*$ for the
single-armed bandit process with subsidy under the discounted reward
criterion pointwise converges to a threshold policy $\pi^*$ as
$\beta_k\rightarrow1$. To see this, we construct $\pi^*$ as follows:
(1) If $m<0$, then the arm is made active all the time; (2) If
$m\ge1$, the arm is made passive all the time; (3) If $0\le m<1$,
then $\omega$ is made passive when current state
$\omega\le\omega^*(m)$, otherwise it is activated. Since
$\omega_{\beta}^*(m)$ converges to $\omega^*(m)$ as
$\beta\rightarrow1$, it is easy to see that $\pi_{\beta_k}^*$
pointwise converges to $\pi^*$ for any $\beta_k\rightarrow1$.
Because the single-armed bandit process with subsidy under the
discounted reward criterion satisfies the value boundedness
condition (see Lemma \ref{lemma:valueboundedness}), the threshold
policy $\pi^*$ is optimal for the single-armed bandit process with
subsidy under the average reward criterion based on Dutta's theorem.

\section*{Appendix F: Proof of Theorem~\ref{thm:whittleindex-timavg}}\label{sec:F}

Since $\omega^*(m)=\lim_{\beta\rightarrow1}\omega^*_{\beta}(m)$ and
$\omega^*_{\beta}(m)$ is monotonically increasing with $m$ (see
Theorem~\ref{thm:indexability}), it is easy to see that
$\omega^*(m)$ is also monotonically increasing with $m$. Therefore,
the bandit is indexable.

Next, we prove that $W(\omega)\defeq \lim_{\beta\rightarrow 1} W_\beta(\omega)$ is
indeed Whittle's index under the average reward criterion. For a belief
state $\omega$ of an arm, its Whittle's index is the infimum subsidy
$m$ such that $\omega$ is in the passive set under the optimal
policy for the arm, \ie the infimum subsidy $m$ such that
$\omega\le\omega^*(m)$ (according to
Lemma~\ref{lemma:thresholdPolicyavg}). From~\eqref{eqn:thresholdavg}
and the monotonicity of $W(\omega)$ with $\omega$, we have that
$W(\omega)$ is the infimum subsidy $m$ such that
$\omega\le\omega^*(m)$.


\section*{Appendix G: Proof of The Structure of Whittle's Index Policy}\label{sec:G}

The proof is an extension of the proof given
in~\cite{Zhao&etal:08TWC} under single-channel sensing ($K=1$).
Consider the belief update of unobserved channels (see
\eqref{eq:omega}).
\begin{equation}
\Tc^1(\omega)=p_{01}+\omega (p_{11}-p_{01}). \label{eq:tau}
\end{equation}
We notice that $\Tc^1(\omega)$ is an increasing function of $\omega$
for $p_{11}>p_{01}$ and a decreasing function of $\omega$ for
$p_{11}<p_{01}$. Furthermore, the belief value $\omega_i(t)$ of
channel $i$ in slot $t$ is bounded between $p_{01}$ and $p_{11}$ for
any $i$ and $t>1$ (see \eqref{eq:omega}).

Consider first $p_{11}\ge p_{01}$. The channels observed to be in
state $1$ in slot $t-1$ will achieve the upper bound $p_{11}$ of the
belief value in slot $t$ while the channels observed to be in state
$0$ the lower bound $p_{01}$. Whittle's index policy, which is
equivalent to the myopic policy, will stay in channels observed to
be in state $1$ and recognize channels observed to be in state $0$
as the least favorite in the next slot. The unobserved channels
maintains the ordering of belief values in every slot due to the
monotonically increasing property of $\Tc^1(\omega)$. The structure
of Whittle index policy for $p_{11}<p_{01}$ can be similarly
obtained by noticing that reversing the order of unobserved channels
in every slot maintains the ordering of belief values due to the
monotonically decreasing property of $\Tc^1(\omega)$.

\section*{Appendix H: Proof of Theorem~\ref{thm:performanceidentical}}\label{sec:H}

The proof for the lower bound of $J_w$ is an extension of that with
single-channel sensing ($K=1$) given in~\cite{Zhao&etal:08TWC}. It
is, however, much more complex to analyze the performance of Whittle's
index policy when $K\ge1$. The lower bound obtained here is looser
than that in~\cite{Zhao&etal:08TWC} when applied to the case of
$K=1$.

Define a transmission period on a channel as the number of
consecutive slots in which the channel has been sensed. Based on the
structure of Whittle index policy, it is easy to show that
\begin{equation}\label{eqn:U-TP}
J_w=\left\{\begin{array}{ll}K(1-\frac{1}{\mbbE[\tau]});~~~&\mbox{if}~~p_{11}\ge
p_{01};\\[0.5em]
K\frac{1}{\mbbE[\tau]};~~~&\mbox{if}~~p_{11}< p_{01}
\end{array}\right.,
\end{equation}
where $\mbbE[\tau]$ is the average length of the transmission period
over the infinite time horizon.

To bound the throughput $J_w$, it is equivalent to bound the average
length of the transmission period $\mbbE[\tau]$ as shown in equation
\eqref{eqn:U-TP}. We consider the following two cases.

\begin{itemize}
\item {\it Case 1: $p_{11}\ge p_{01}$}
\end{itemize}

Let $\omega$ denote the belief value of the chosen channel in the
first slot of a transmission period. The length $\tau(\omega)$ of
this transmission period has the following distribution.

\begin{equation}\label{eq:TPomega}
\Pr[\tau(\omega)=l]= \left\{\begin{array}{ll} 1-\omega, & l = 1 \\
\omega p_{11}^{l-2}p_{10}, & l>1\\
\end{array}
\right..
\end{equation}
It is easy to see that if $\omega'\ge \omega$, then $\tau(\omega')$
stochastically dominates $\tau(\omega)$.

From the structure of Whittle index policy, $\omega=\Tc^k(p_{01})$,
where $k$ is the number of consecutive slots in which the channel
has been unobserved since the last visit to this channel. When the
user leaves one channel, this channel has the lowest priority. It
will take at least $\lfloor\frac{N-K}{K}\rfloor$ slots before the
user returns to the same channel, \ie $k\ge
\lfloor\frac{N}{K}\rfloor-1$. Based on the monotonically increasing
property of the $k$-step transition probability $\Tc^k(p_{01})$ (see
Fig.~\ref{fig:ksteppositive}), we have
$\omega=\Tc^k(p_{01})\ge\Tc^{\lfloor\frac{1}{\alpha}\rfloor-1}(p_{01})$.
Thus $\tau(\Tc^{\lfloor\frac{N}{K}\rfloor-1}(p_{01}))$ is
stochastically dominated by $\tau(\omega)$, and the expectation of
the former leads to the lower bound of $J_w$ given
in~\eqref{eqn:optimalbound}.

\begin{itemize}
\item {\it Case 2: $p_{11}< p_{01}$}
\end{itemize}

In this case, $\tau(\omega)$ has the following distribution:

\begin{equation}\label{eq:TPomega1}
\Pr[\tau(\omega)=l]= \left\{\begin{array}{ll} \omega, & l = 1 \\
(1-\omega) p_{00}^{l-2}p_{01}, & l>1\\
\end{array}
\right..
\end{equation}
Opposite to case 1, $\tau(\omega')$ stochastically dominates
$\tau(\omega)$ if $\omega'\le\omega$ .

From the structure of Whittle's index policy,
$\omega=\Tc^k(p_{11})$, where $k$ is the number of consecutive slots
in which the channel has been unobserved since the last visit to
this channel. If $k$ is odd, then
$\Tc^k(p_{11})\ge\Tc^{2\lfloor\frac{N}{K}\rfloor-2}(p_{11})$ since
$2\lfloor\frac{N}{K}\rfloor-2$ is an even number (see
Fig.~\ref{fig:kstepnegative}). If $k$ is even, then $k$ is at least
$2\lfloor\frac{N-K}{K}\rfloor$. we have $\omega=\Tc^k(p_{11})\ge
\Tc^{2\lfloor\frac{N}{K}\rfloor-2}(p_{11})$. Thus $\tau(\omega)$ is
stochastically dominated by
$\tau(\Tc^{2\lfloor\frac{N}{K}\rfloor-2}(p_{11}))$, and the
expectation of the latter leads to the lower bound of $J_w$ as given
in~\eqref{eqn:optimalbound1}.

Next, we show the upper bound of $J$. From~\eqref{eqn:relaxavg},
we have $J\le\inf_{m}\{NJ_m-m(N-K)\}$ since channels are
stochastically identical.

When $p_{11}\ge p_{01}$, we have
\begin{eqnarray}\label{eqn:j1}
J\le\min_{m\in\{\frac{\omega_o}{1-p_{11}+\omega_o},~
0\}}NJ_m-m(N-K)= \min\{\frac{K\omega_o}{1-p_{11}+\omega_o},
N\omega_o\}.
\end{eqnarray}
When $p_{11}>p_{01}$, we have
\begin{eqnarray}\label{eqn:j2}
J\le\min_{m\in\{\frac{p_{01}}{1-\Tc^1(p_{11})+p_{01}},~0\}}NJ_m-m(N-K)=\min\{\frac{Kp_{01}}{1-\Tc^1(p_{11})+p_{01}},
N\omega_o\}.
\end{eqnarray}

\section*{Appendix I: Proof of Corollary~\ref{cor:appfac}}\label{sec:I}

It has been shown that the myopic policy is optimal when $K=1$
and $p_{11}\ge p_{01}$~\cite{Zhao&etal:08TWC,Ahmad&etal} (note that
for $N=2,3$ negatively correlated channels, the optimality of the myopic policy has also been
established). Based on the equivalency between Whittle's index policy
and the myopic policy, we conclude that Whittle index policy is optimal
for $K=1$ and $p_{11}\ge p_{01}$.

We now prove that Whittle's index policy is optimal when $K=N-1$.
We construct a genie-aided system where the user knows the states
$S_i(t)$ of all channels at the end of each slot $t$. In this
system, Whittle's index policy is clearly optimal, and the optimal
performance is the upper bound of the original one. For the original
system where the user only knows the states of the sensed $N-1$
channels, we notice that the channel ordering by Whittle's index
policy in each slot is the same as that in the genie-aided system.
Whittle's index policy thus achieves the same performance as in the
genie-aided system. It is thus optimal.

According to Theorem~\ref{thm:performanceidentical}, we arrive at
the following inequalities (notice that $J_w\ge K\omega_o$).
\begin{equation}\label{eqn:U-TP1}
\eta\ge\left\{\begin{array}{ll}\max\{1-p_{11}+\omega_o,
\frac{K}{N}\},~~~&\mbox{if}~~p_{11}\ge
p_{01}\\[0.5em]
\max\{\frac{1-\Tc^1(p_{11})+p_{01}}{1-p_{11}+p_{01}},
\frac{K}{N}\},~~~&\mbox{if}~~p_{11}< p_{01}
\end{array}\right..
\end{equation}
From~\eqref{eqn:U-TP1}, we have $\eta\ge\frac{K}{N}$.

Next, we show that Whittle's index policy achieves at least half the
optimal performance for negatively correlated channels
($p_{11}<p_{01}$). In this case, we have
\[\eta\ge\frac{1-\Tc^1(p_{11})+p_{01}}{1-p_{11}+p_{01}}=1+\frac{(p_{11}-p_{01})(1-p_{11})}{1-(p_{11}-p_{01})}
\ge 1-\frac{(1-p_{11})^2}{2-p_{11}}\ge 0.5.\]

\bibliographystyle{ieeetr}

\end{document}